\documentclass[aps,prl,twocolumn,superscriptaddress,notitlepage,noshowkeys]{revtex4-2}

\usepackage[utf8]{inputenc}
\usepackage[T1]{fontenc}
\usepackage[english]{babel}
\usepackage{csquotes}
\usepackage{tensor}
\usepackage{setspace}
\usepackage{mathtools,amssymb,amsthm} 
\usepackage{array}
\usepackage[new]{old-arrows}
\usepackage{todonotes}
\usepackage{scalerel,stackengine}
\usepackage{wrapfig}
\usepackage[shortlabels]{enumitem}
\usepackage{fancyhdr}
\usepackage{graphicx}
\usepackage{chngcntr}
\usepackage{dsfont}
\usepackage{mathrsfs}
\usepackage{xcolor}
\usepackage{url}
\usepackage[colorlinks=true, allcolors=darkblue]{hyperref}
\usepackage[capitalize]{cleveref}
\usepackage{tikz}
\usepackage{pgfplots}

\stackMath

\definecolor{darkblue}{rgb}{0.1, 0.1, 0.6}

\definecolor{1}{RGB}{31, 119, 180}
\definecolor{2}{RGB}{255, 127, 14}
\definecolor{3}{RGB}{44, 160, 44}
\definecolor{4}{RGB}{148, 103, 187}
\definecolor{5}{RGB}{140, 86, 75}

\usepackage{pgfkeys}
    
    \def\addlegendimage{\csname pgfplots@addlegendimage\endcsname}

\theoremstyle{definition}
\newtheorem{thm}{Theorem}
\newtheorem{lem}[thm]{Lemma}

\newtheorem{prop}[thm]{Proposition}
\newtheorem{main}[thm]{Main~Result}
\newtheorem*{rem*}{Remark}
\newtheorem*{ex*}{Example}
\newtheorem*{defin*}{Definition}
\newtheorem*{thm*}{Theorem}
\newtheorem*{lem*}{Lemma}
\newtheorem*{cor*}{Corollary}
\newtheorem*{def*}{Definition}
\newtheorem*{main_one}{First Main Result}
\newtheorem*{main_two}{Second Main Result}

\DeclarePairedDelimiter{\rb}{\lparen}{\rparen}

\DeclarePairedDelimiter{\abs}{\lvert}{\rvert}

\DeclarePairedDelimiterX{\norm}[1]{\lVert}{\rVert}{\ifblank{#1}{\placeholder}{#1}}
\DeclarePairedDelimiterX{\ip}[2]{\langle}{\rangle}{#1 , #2}
\providecommand\given{}
\newcommand{\SetSymbol}[1][]{%
\nonscript\ #1\vert
\allowbreak
\nonscript\
\mathopen{}
}

\DeclarePairedDelimiterX{\Set}[1]\{\}{%
\renewcommand\given{\SetSymbol[\delimsize]}
#1
}

\newcommand{\id}{\operatorname{id}}

\newcommand{\slim}{\operatorname*{s-lim}}

\renewcommand{\Im}{\operatorname{Im}}

\newcommand{\conj}[1]{\overline{#1}}

\newcommand{\bra}[1]{\langle #1|}
\newcommand{\ket}[1]{|#1\rangle}

\newcommand{\supind}[1]{^{\scaleobj{0.8}{(#1)}}}

\newcommand{\emp}[2][b]{\if b#1\textbf{#2}\else\textit{#2}\fi}

\renewcommand{\qquad}{\hspace{60pt}}

\newcommand{\sref}[2]{\hyperref[#2]{#1 \ref*{#2}}}

\newcommand{\rme}{\mathrm{e}}
\newcommand{\rmi}{\mathrm{i}}

\newcommand{\1}{\mathds{1}}

\newcommand{\CCC}{\mathscr{C}}

\newcommand{\D}{\mathcal{D}}
\newcommand{\DD}{\mathcal{D}}
\newcommand{\DDD}{\mathscr{D}}
\renewcommand{\d}{\ \textrm{d}}

\renewcommand{\H}{\mathcal{H}}

\newcommand{\NN}{\mathbb{N}}

\newcommand{\RR}{\mathbb{R}}
\newcommand{\RRR}{\mathscr{R}}

\newcommand{\T}{\mathcal{T}}

\newcommand{\V}{\mathcal{V}}

\newcommand{\W}{\mathcal{W}}

\newcommand\dom[1]{\mathcal{D}(#1)}
\def\placeholder{{\hspace{1pt}\,{\mathbin{\vcenter{\hbox{\scalebox{0.5}{$\bullet$}}}}}\hspace{1pt}\,}}
\renewcommand\Im{\mathrm{Im}\,}

\makeatletter
\renewcommand{\p@subsection}{}
\makeatother

\begin{document}

	\title{State-dependent Trotter Limits and their approximations}
	\author{Daniel Burgarth}
	\affiliation{Center for Engineered Quantum Systems, Macquarie University, 2109 NSW, Australia}
	\affiliation{Physics Department, Friedrich-Alexander Universit{\"a}t of Erlangen-Nuremberg, Staudtstr. 7, 91058 Erlangen, Germany}
	\author{Niklas Galke}
	\affiliation{F\'{\i}sica Te\`{o}rica: Informaci\'{o} i Fen\`{o}mens Qu\`{a}ntics, Departament de F\'{\i}sica, Universitat Aut\`{o}noma de Barcelona, 08193 Bellaterra, Spain}
	\author{Alexander Hahn}
	\affiliation{Center for Engineered Quantum Systems, Macquarie University, 2109 NSW, Australia}
	\author{Lauritz van Luijk}
	\affiliation{Institut f{\"u}r Theoretische Physik, Leibniz Universit{\"a}t Hannover, Appelstra{\ss}e 2, 30167 Hannover, Germany}

	\begin{abstract}
		In this paper, we study how to discretize the Trotter product formula of operators with continuous degrees of freedom -- such as the position of particles in molecules, or the amplitude of electromagnetic fields -- to make them amenable to digital simulations. Usually, such systems are simulated numerically via the Trotter product formula of their discretized approximations, but this approach can potentially lead to fallacious results. Here, we find sufficient conditions to conclude the validity of this approximate discretized physics. We develop a scheme to verify numerical results involving Trotterization of truncated operators. Essentially, it depends on the state-dependent Trotter error of the latter, for which we establish explicit bounds. These bounds are of independent interest and may also find applications in quantum chemistry.
	\end{abstract}
	
	\maketitle

	\emph{Introduction.}---The Trotter product formula is a central tool in the study of quantum dynamics with a vast amount of applications, such as quantum computing \cite{Feynman1982,Lloyd1996,Poulin2015,Childs2021}, quantum field theory \cite{Klauber2013,Simon2005,Nicola2019,Gaveau2004} and its path integral formulation \cite{Johnson2002}, quantum control \cite{Misra1977,Arenz2018,Arenz2017}, open quantum systems \cite{Rivas2012}, quantum optics \cite{Fitzek2020}, Floquet dynamics \cite{Sieberer2019,Chinni2022,Kargi2021} and quantum many body systems \cite{Verstraete2004,Zwolak2004,Vidal2004,Cincio2008}. It allows to split the time-evolution of a quantum system, that is hard or cost-expensive to implement, into simpler components. Due to this feature, it has become the standard contraption for numerically simulating the dynamics of large quantum systems.
	Nevertheless, two important questions regarding these simulations stay open. On one hand, it is not clear how the error due to the Trotter approximation depends on the input state of the system. For example, one would expect that the approximation becomes better when only low-energy states are considered.
	On the other hand, digital computers -- no matter if classical or quantum -- can only simulate a finite number of levels. However, many relevant physical systems are actually represented by continuous degrees of freedom, i.e.\ infinite dimensions. Therefore, the common practice is to truncate the Hilbert space on some finite level \cite{Lambert2019,Jeckelmann1998,Kargi2021,Vidal2007} and deal with the finite-dimensional truncated operators. These truncations, however, are only approximations to the full system. Hence in total, this simulation procedure involves two simultaneous approximations: First, approximating the time-evolution of the full quantum system by some finite-dimensional truncations; second, approximating the full time-evolution of the quantum system by the Trotter product formula. We summarize the interplay of the two approximations in Fig.~\ref{fig:sketch}.
	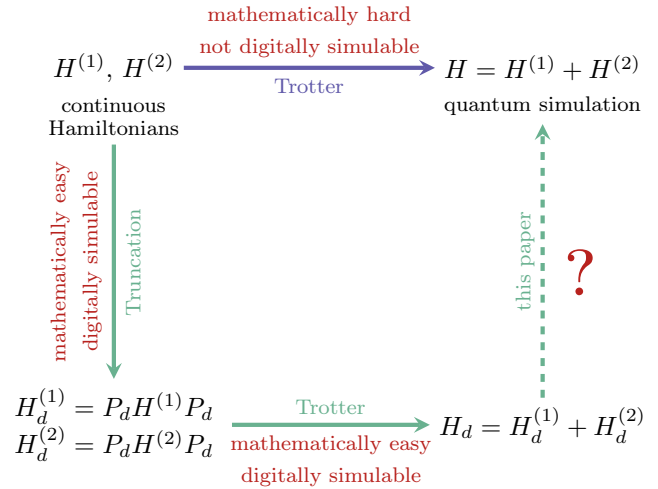
\begin{figure}
		\begin{center}
			\begin{tikzpicture}[node distance=0.5cm]
				\node[anchor=center] (H12) at (0,0) {$H^{(1)}$, $H^{(2)}$};
				\node[below=0.5cm of H12.center, anchor=center] (textH121) {\footnotesize continuous};
				\node[below=0.3cm of textH121.center, anchor=center] (textH122) {\footnotesize Hamiltonians};
				\node[anchor=center] (H) at (5.7,0) {$H=H^{(1)}+H^{(2)}$};
				\node[below=0.5cm of H.center, anchor=center] (textH) {\footnotesize quantum simulation};
				\node[anchor=center] (H1n) at (0,-4.5) {$H_d^{(1)}=P_dH^{(1)}P_d$};
				\node[below=0.5cm of H1n.center, anchor=center] (H2n) {$H_d^{(2)}=P_dH^{(2)}P_d$};
				\node[below=0.25cm of H1n.east, anchor=center] (H12n) {};
				\node[anchor=center] (Hn) at (5.7,-4.75) {$H_d=H_d^{(1)}+H_d^{(2)}$};
				\draw[-stealth,color=4,ultra thick] (H12) -- node[below] {\footnotesize Trotter} node[above,align=center,color=2] {\footnotesize mathematically hard\\\footnotesize not digitally simulable} ++ (H);
				\draw[-stealth,color=3,ultra thick] (textH122) -- node[rotate=90,right=of H1n,anchor=south] {\footnotesize Truncation} node[rotate=90,left=of H1n,anchor=center,align=center,color=2] {\footnotesize mathematically easy\\\footnotesize digitally simulable} ++ (H1n);
				\draw[-stealth,color=3,ultra thick] (H12n) -- node[below,align=center,color=2] {\footnotesize mathematically easy\\\footnotesize digitally simulable} node[above] {\footnotesize Trotter} ++ (Hn);
				\draw[-stealth,color=3,ultra thick,dashed] (Hn) -- node[rotate=90,left=of Hn,anchor=north] {\footnotesize this paper} ++ (textH);
				\node[color=2] (qm) at (6.2,-2.725) {\Huge \bf{?}};
			\end{tikzpicture}
		\end{center}
		\vspace{-1\baselineskip}
		\caption{In the simulation of a quantum system one starts with two simple continuous-degree Hamiltonians $H^{(1)}$ and $H^{(2)}$. Then, the goal is to simulate the dynamics under $H=H^{(1)}+H^{(2)}$ through the Trotter product formula (blue path). However, due to its continuous nature, it is not possible to implement this strategy on digital simulators such as computers. Therefore, it is common practice to follow the green path instead. This is, one first truncates $H^{(1)}$ and $H^{(2)}$ to a finite level $d$. Afterwards, one implements the dynamics under the sum of the truncated Hamiltonians $H_d=H_d^{(1)}+H_d^{(2)}$ via the Trotter product formula. Unfortunately, $H_d$ does not necessarily cover the physics of the full model $H$ correctly, so that this method can lead to deceptive results. Here, we establish a method to close this gap, which is indicated by the dashed green arrow.} \label{fig:sketch}
		\vspace{-1\baselineskip}
	\end{figure}
	A priori, it is not self-evident whether they are compatible with each other. In fact, we will discuss a simple and well-known example, where it seems that they are not. In this context, the main obstacle is the question of convergence, i.e.: Is the infinite-dimensional Trotter product formula a valid approximation? Whilst it always converges in finite dimensions \cite{Childs2021,Suzuki1985}, it only converges under certain conditions in the infinite-dimensional case \cite{Kato1978,Lapidus1981}. Therefore, a rigorous numerical treatment of infinite-dimensional truncated operators would need to additionally deliver a reasoning why the Trotterization and truncation procedures are nonconflicting.
	Such a discussion about the validity of the simulations is missing in the literature, making numerical results not fully conclusive.
	
	In this paper, we address both of the aforementioned questions and show that they are closely related. That is, we first establish state-dependent bounds for the Trotter error. Afterwards, we use them to provide sufficient conditions on the truncations, which ensure the validity of simultaneous Trotterization. Firstly, as a technical condition, the Hamiltonians of consideration have to be jointly approximable by the \emph{same} truncation scheme. Secondly, as one would expect intuitively, the state-dependent error of the finite-dimensional Trotter approximation has to eventually saturate in the truncation dimension. In this case, our method allows to infer bounds for the Trotter error in infinite dimensions from the state-dependent ones of their finite-dimensional truncations. These results have applications to a wide class of physical models. Not only do even the simplest textbook examples such as the quantum harmonic oscillator and the particle in a box have continuous degrees. But also e.g.\ in quantum optics, the Rabi-model \cite{Rabi1936,Rabi1937} and the spin-boson model \cite{Jaynes1963}, which describe the interaction of matter with light, are infinite-dimensional Hamiltonians. In quantum chemistry, atoms and molecules have continuous positions and momenta. The simulation of their chemical behaviours has been proposed as an impactful application of the emerging field of quantum simulation \cite{Feynman1982,Lloyd1996}. In addition, modern quantum information makes heavy use of infinite-dimensional systems, such as GKP qubits \cite{Gottesman2001} or superconducting transmon qubits \cite{Koch2007,Schreier2008}, which were the platform for the ``quantum supremacy'' experiment \cite{Arute2019}.
	
	Additionally, our state-dependent Trotter bound is of independent and broader interest in the field of quantum information and computation beyond the scope of this paper. Recently, such state-dependent Trotter error bounds enjoy a considerable interest in particular in the area of Hamiltonian simulation. For instance in \cite{Sahinoglu2021,Yi2022}, the asymptotic scaling of the finite-dimensional Trotter error is derived when the input state is only supported in a (low-energy) subspace. These results rely on strong assumptions on the respective Hamiltonians in order to bound the leakage from this (low-energy) subspace. The authors in \cite{Becker2021, Luijk2022} study energy-constrained distances between time evolutions generated by Hamiltonians with continuous degrees of freedom. Their results have applications in the realm of Trotterization but only hold for special Hamiltonians generating a particular type of dynamics. In addition, in \cite{Jahnke2000,An2021} state-dependent Trotter error bounds are derived for certain (time-dependent) Hamiltonians. These bounds rely on restrictive assumptions on  the structure of the commutator applied to the input state. The advantage of our bound, compared to the existing bounds in the literature, is its validity \emph{for all} finite-dimensional Hamiltonians. Furthermore, it is an \emph{explicit} bound, which shows a particularly simple dependence on the input state. By taking linear combinations of input vectors, we are immediately able to bound the Trotter convergence speed in arbitrary subspaces.

	\emph{Setting.}---Before turning to our main results, let us briefly summarise the setting, in which the Trotter product formula and its truncated versions are studied. A more elaborate introduction is given in the Supplementary Material, Sec.~\hyperlink{app:preliminaries}{A}. We will consider two Hamiltonians $H^{(1)}$ and $H^{(2)}$ acting on an infinite-dimensional Hilbert space $\H$ with countable basis. This is, two self-adjoint operators, which are defined on a dense domain of input states $\D(H^{(1)})$, $\D(H^{(2)})$, respectively. The domains encode important features of the physical system, such as boundary conditions. It often suffices to look at a core, which is a particular subspace of the domain that already carries all the relevant information. See the Supplementary Material, Sec.~\hyperlink{app:preliminaries}{A} for details. We do not make any assumptions on the spectrum of the Hamiltonians in order to accommodate for both purely continuous and (infinitely many) discrete degrees of freedom. The time-evolution under each Hamiltonian is given by a unitary operator $U^{(1)}(t)=\rme^{-\rmi tH^{(1)}}$, $U^{(2)}(t)=\rme^{-\rmi tH^{(2)}}$, respectively. Now, the Trotter product formula, \emph{if} converging, approximates the dynamics $U(t)=\rme^{-\rmi tH}$ of a self-adjoint operator $H$ by the $n$-fold product of the individual unitaries $X(t/n)^n\equiv[U^{(1)}(t/n)U^{(2)}(t/n)]^n$. On the common core $H$ restricts to the sum of $H^{(1)}$ and $H^{(2)}$, see the Supplementary Material, Sec.~\hyperlink{app:proof}{C}.  In order to construct finite-dimensional approximations of the above operators, we will first fix a basis $\Set{\ket{j}:j\in\NN_0}$ of the Hilbert space $\H$. Then, we obtain a truncated, finite-dimensional Hilbert space $V_d = \mathrm{span}\Set{\ket 0,\dots,\ket{d-1}}$ of dimension $d<\infty$ simply by gathering the first $d$ basis vectors. We can project any vector in $\H$ onto $V_d$ with $P_d=\sum_{j=0}^{d-1}\ket{j}\bra{j}$. Similarly, for our Hamiltonians we receive their truncated versions by $H_d^{(1)}=P_dH^{(1)}P_d$, $H_d^{(2)}=P_dH^{(2)}P_d$ and $H_d=P_dHP_d$, if existing. All the truncated Hamiltonians only act on the finite-dimensional $V_d$ and can, therefore, be represented by Hermitian $d\times d$ matrices. They generate finite-dimensional unitary dynamics on $V_d$ via $U_d^{(1)}(t)=\rme^{-\rmi t H_d^{(1)}}$, $U_d^{(2)}(t)=\rme^{-\rmi t H_d^{(2)}}$ and $U_d(t)=\rme^{-\rmi t H_d}$, respectively.
	One goal of this paper is to give conditions on the finite-dimensional truncated Trotter product formula which imply the convergence of the full infinite-dimensional case. We will explain in the Supplementary Material, Sec.~\hyperlink{app:preliminaries}{A} that it is necessary to include the dependence on an input state (also see \cite[Sec.~3]{Ichinose2004} and \cite[Sec.~VIII.8]{ReedSimon1981}). That is, by studying $[U_d^{(1)}(t/n)U_d^{(2)}(t/n)]^n\ket{\psi_d}\xrightarrow{n\to\infty} U_d(t)\ket{\psi_d}$, we would like to draw conclusions about $[U^{(1)}(t/n)U^{(2)}(t/n)]^n\ket{\psi}\xrightarrow{n\to\infty} U(t)\ket{\psi}$, where $\ket{\psi_d}=P_d\ket{\psi}$. To this end, we introduce the state-dependent Trotter error in the finite-dimensional case
	\begin{align}
		&b_d^{(n)}(\ket{\psi};t):=\nonumber\\&\big\Vert \big([U_d^{(1)}(t/n)U_d^{(2)}(t/n)]^n-U_d(t)\big)\ket{\psi_d}\big\Vert,
	\end{align}
	from which we can inherit a dimension-independent Trotter error by taking the limit $d\to\infty$
	\begin{equation}
		b^{(n)}(\ket{\psi} ; t) := \limsup_{d\to\infty}  b_d\supind n(\ket{\psi};t).\label{eq:Trotter_error_infinite}
	\end{equation}
	Notice that $b^{(n)}(\ket{\psi} ; t)$ might reach the maximum distance of $2$ even though $b_d^{(n)}(\ket{\psi};t)$ is considerably smaller for a relatively large $d$.

	\emph{Main Results.}---Our first main result is an explicit state-dependent Trotter error bound for arbitrary \emph{finite-dimensional} systems. Usually, when computing the norm Trotter error, one would reduce the difference $X(t/n)^n-U(t)$ to one period $X(t/n)-U(t/n)$ by using a telescope sum identity, see for instance \cite{Suzuki1985,Childs2021}. This approach introduces additional unitaries, which can be removed through sub-multiplicativity of the matrix norm. However in the state-dependent case, these unitaries lead to a non-trivial rotation of the input state, which hinges on the total evolution time $t$. For this reason, we need to follow a conceptionally different approach than all other proofs of Trotter convergence in the literature. Our proof is inspired by \cite[Lemma~1]{Burgarth2022}. We treat the generator of the Trotterized evolution $X(t/n)^n$ as a $\frac{2t}{n}$-periodic time-dependent Hamiltonian $\tilde{H}(s)$, which is repeatedly switching between $H^{(1)}$ and $H^{(2)}$ after a time $s=\frac{t}{n}$. Then, the distance between the Trotterized and the target evolution is related to an integral action $S_{21}(\tau)=\int_0^\tau \frac{1}{2}H-\tilde{H}(s)\mathrm{d}s$, which vanishes after a full cycle, $S_{21}(2t)=0$. Thus, applied to an eigenstate of $H$, the integral action can essentially be bounded by the maximum of the distance between $H^{(1)}$ or $H^{(2)}$ and the average generator $\frac{1}{2}H$. See the Supplementary Material, Sec.~\hyperlink{app:proof_bound}{B}. In total, we obtain the following bound.

	\begin{main}[Finite-dimensional state-dependent Trotter error bound]\label{thm:finite_bounds}
		Consider the case of the \emph{finite-dimensional} Trotter formula, i.e.\ the two Hamiltonians $H^{(1)}$ and $H^{(2)}$ are acting on a finite-dimensional Hilbert space $\H$. If $\ket{\varphi}\in\H$ is an eigenstate of $H^{(1)}+H^{(2)}$ according to the eigenvalue equation $(H^{(1)}+H^{(2)})\ket{\varphi}=h\ket{\varphi}$, then the Trotter error is bounded by
		\begin{align}
			&\left\Vert\left([U^{(1)}(t/n)U^{(2)}(t/n)]^n-U(t)\right)\ket{\varphi}\right\Vert\nonumber\\&\leq
			\frac{t^2}{2n}\left(\Big\Vert \Big[H^{(1)}-\frac{h}{2}\Big]^2\ket{\varphi}\Big\Vert + \Big\Vert \Big[H^{(2)}-\frac{h}{2}\Big]^2\ket{\varphi}\Big\Vert\right).
		\end{align}
	\end{main}
	
	The usual error bounds in a matrix norm essentially quantify a worst-case scenario. Instead, this state-dependent bound still accounts for the relationship between the particular input state $\ket{\varphi}$ and the fidelity of Trotterization. Therefore, it can lead to a tighter scaling for certain $\ket{\varphi}$. For instance, this is crucial when considering infinite-dimensional Trotter problems. Approximating those with increasing cutoff dimensions $d$ adds more and more states to the Hilbert space $V_d$. In turn, the worst-case Trotter error grows with the truncation dimension. Since one is ultimately interested in the limit $d\to\infty$, this is a rather unfavourable property and state-dependent error bounds are mandatory. We study these truncated product formulas in our next main result. It manifests a method to check convergence of the Trotter product formula which is particularly well-suited to numerical simulations. Intuitively, if one seeks out that the state-dependent Trotter error $b_d^{(n)}(\ket{\psi};t)$ saturates with increasing $d$, one would expect that the infinite-dimensional analogue converges for this particular $\ket{\psi}$. Indeed, this is what we find.

	\begin{main}[Truncation of an infinite-dimensional Trotter product]\label{thm:mainthm}
		The \emph{infinite-dimensional} Trotter product formula converges for all $\ket{\psi}\in\H$,
		\begin{equation}
			\rb*{U\supind1(t/n)U\supind2(t/n)}^n\ket{\psi} \xrightarrow{n\to\infty} U(t)\ket{\psi},
		\end{equation}
		if the following two conditions are satisfied:
		\begin{enumerate}[(i)]
			\item $H^{(1)}$ and $H^{(2)}$ can be simultaneously approximated with $H_d^{(1)}$ and $H_d^{(2)}$ by the same truncation scheme. That is, $\V=\bigcup_d V_d$ is a common core of $H^{(1)}$ and $H^{(2)}$.
			\item For all total evolution times $t\in\RR$, the dimension-independent Trotter error  goes to zero, $b^{(n)}(\ket{\psi};t) \to 0$, as $n\to\infty$.
		\end{enumerate}
		In this case, $U(t)$ gives a unitary dynamics, whose generator $H$ is self-adjoint and agrees with $(H\supind1+H\supind2)$, wherever both are defined. A generalized version to non-unitary dynamics can be found in the Supplementary Material, Sec.~\hyperlink{app:proof}{C}, see Thm.~\ref{thm:main}.
	\end{main}
	
In order to prove this, we first show that $U_d(t)\ket{\psi}$ has a limit as $d\to\infty$, i.e.\ the finite-dimensional approximations applied to a state $\ket{\psi}$ converge to a well-defined $\ket{\phi(t)}$. 
Since the Trotter limit $n\to\infty$ is always well-defined in finite dimensions \cite{Suzuki1985}, the $U_d(t)$ can be obtained by Trotterization. 
Then, we prove that the infinite-dimensional Trotter formula, $X(t/n)^n\ket{\psi}\to\ket{\phi(t)}$ as $n\to\infty$, converges to the same limit $\ket{\phi(t)}$.
The last step is to show that this limit is indeed governed by a unitary time-evolution, whose generating Hamiltonian agrees with $H^{(1)}+H^{(2)}$ on $\mathcal V$  These ramifications are derived using a modification of the \emph{Trotter-Kato approximation theorems} \cite[Thm.~4.8, Thm.~4.9]{EN01} inspired by \cite[Thm.~2.3]{Duffield1992}. Notice that we only have $\V\subsetneq\H$, thus in particular $\V\neq\H$. Furthermore, two Hamiltonians may have a non-overlapping domains, see e.g.\ \cite[Example~3.8]{Arenz2018}. This situation is excluded by our common core assumption.

	Our truncation main result establishes a way to obtain explicit error bounds for the infinite-dimensional Trotter product formula once the finite-dimensional case is under control: A simple calculation shows that (see the Supplementary Material, Sec.~\hyperlink{app:proof}{C})
	\begin{align}
		b^{(n)}(\ket{\psi};t)
		\leq \rb[\bigg]{\sum_{j=1}^d \big|\langle j | \psi \rangle\big|^2 b^{(n)}(\ket j;t)^2 }^{1/2}.\label{eq:error_bound_infinite}
	\end{align}
	Thus, in order to conclude that the Trotter product formula converges for all $\ket{\psi}\in\H$, it suffices to require $b^{(n)}\rb*{\ket{j}; t}\to 0$ as $n \to\infty$ for all basis vectors $\ket{j}$.
	This is because if $\ket{\psi} \in \V$ -- and hence $\ket{\psi} \in V_d$ for sufficiently large $d$ -- then Eq.~\eqref{eq:error_bound_infinite} also goes to zero.

	\emph{Examples}.---Let us numerically study an example of a convergent and a non-convergent Trotter scenario. We denote the position operator with $Q$ and the momentum operator with $P$. The truncation is done in the Fock basis $\Set{\ket{m}}$, which is a common core of all the considered operators. That is, we project with $P_d=\sum_{m=0}^{d-1}\ket{m}\bra{m}$ onto the finite-dimensional Hilbert space $V_d$. Fig.~\ref{fig:strong_HO_squeezing_small} shows the case of $H^{(1)}=\frac{1}{2}\left(Q^2+P^2\right)$, i.e.\ the quantum harmonic oscillator, and $H^{(2)}=\frac{1}{2}\left(QP+PQ\right)$, i.e.\ the Hamiltonian which generates the squeezing transformation. It is known that this infinite-dimensional Trotter problem converges for all $\ket{\psi}\in\H$, see for instance \cite{Lapidus1981}. Indeed, we numerically find that the Trotter error on the finite-dimensional truncations $V_d$ is bounded independent of the dimension of truncation for all considered Fock basis states $\ket{m}$. For a non-convergent Trotter scenario, we look at the two operators $H^{(1)}=Q^3$ and $H^{(2)}=P^2$ in Fig.~\ref{fig:strong_Q3_P2_samll}. The corresponding sum $H^{(1)}+H^{(2)}=P^2+Q^3$ describes a particle in a $Q^3$ potential. This Trotter problem does not converge as even in the classical version of this Hamiltonian, the particle would escape to infinity at finite times \cite{Zhu1993}. Our numerical study also reveals that the Trotter error on the finite-dimensional truncations $V_d$ does not saturate in the truncation dimension $d$. Instead, we observe a phase transition to a chaotic behaviour similar to \cite{Sieberer2019,Kargi2021}.
	
	Through our state-dependent error bound, we can also perform an analytical examination in certain cases. For this, consider $H^{(1)}=\frac{1}{2}Q^2$ and $H^{(2)}=\frac{1}{2}P^2$ and truncate at level $d$ in the Fock basis $\Set{\ket{m}}$ as before. We then find the analytic bound
	\begin{align}
		b_d^{(n)}&(\ket{m};t)\nonumber\\&\leq
		\frac{t^{2}}{4n}\sqrt{\frac{3}{8} \left(m (m+1) \left(m^2+m+14\right)+10\right)},\label{Trotter_X2P2}
	\end{align} 
	see the Supplementary Material, Sec.~\hyperlink{app:example}{D}. This bound does not depend on the cutoff dimension $d$. Hence, by our truncation main result the full continuous Trotter problem converges and its Trotter error can be quantified for Fock states explicitly by Eq.~\eqref{Trotter_X2P2}. For a numerical comparison, see Fig.~\ref{fig:strong_Q2_P2_small}.

	\begin{figure}
		\begin{tikzpicture}[mark size={0.5mm}, scale=1]
			\pgfplotsset{%
				width=.5\textwidth,
				height=.35\textwidth
			}
			\begin{axis}[
				ylabel near ticks,
				xlabel={Dimension of truncation $d$},
				ylabel={Trotter error $b_d^{(n)}(\ket{m};t)$},
				x post scale=1,
				y post scale=1,
				legend pos=north west,
				legend cell align={left},
				legend columns=2,
				label style={font=\footnotesize},
				tick label style={font=\footnotesize},
				ytick={0,0.05,0.1,0.15},
				yticklabels={0,0.05,0.1,0.15},
				]
				\addplot[color=1, only marks] table[x=n, y=1, col sep=comma]{StrongHOsq_small.csv};
				\addplot[color=2, only marks] table[x=n, y=2, col sep=comma]{StrongHOsq_small.csv};
				\addplot[color=3, only marks] table[x=n, y=3, col sep=comma]{StrongHOsq_small.csv};
				\addplot[color=4, only marks] table[x=n, y=4, col sep=comma]{StrongHOsq_small.csv};
				\addplot[color=5, only marks] table[x=n, y=5, col sep=comma]{StrongHOsq_small.csv};
				\legend{\footnotesize$\ket{0}$, \footnotesize$\ket{1}$, \footnotesize$\ket{2}$, \footnotesize$\ket{3}$, \footnotesize$\ket{4}$};
			\end{axis}
		\end{tikzpicture}
		\vspace{-0.5\baselineskip}
		\caption{State-dependent Trotter error for the operators $H^{(1)}=\frac{1}{2}\left(Q^2+P^2\right)$ and $H^{(2)}=\frac{1}{2}\left(QP+PQ\right)$. We consider the first five Fock-basis states $\Set{\ket{m}}=\Set{\ket{0},\dots,\ket{4}}$ for different dimensions of truncations $d=1,\dots,300$. The total evolution time is fixed to $t=3$ and the number of Trotter steps is $n=1000$. The saturation of the error indicates the convergence of this Trotter problem.}
		\label{fig:strong_HO_squeezing_small}
	\end{figure}
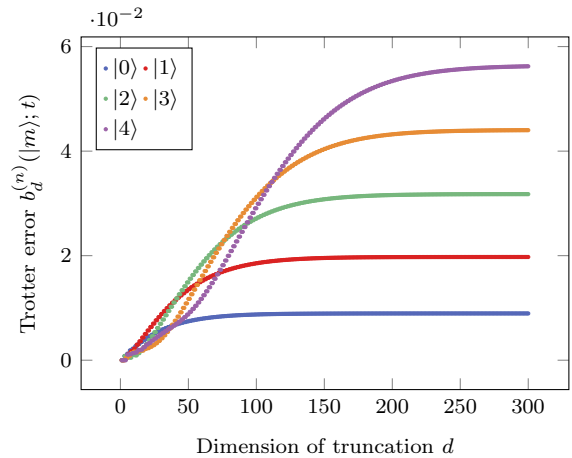
	
	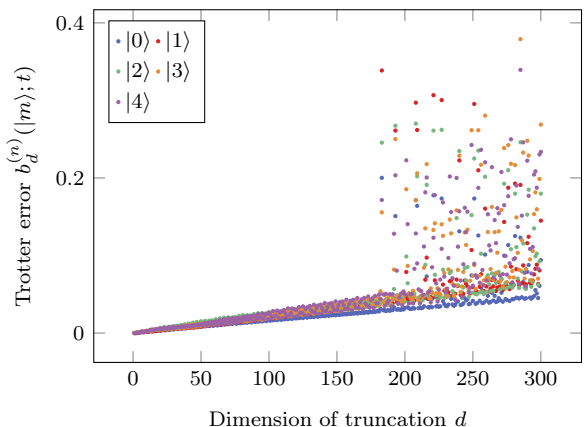
\begin{figure}
		\begin{tikzpicture}[mark size={0.5mm}, scale=1]
			\pgfplotsset{%
				width=.5\textwidth,
				height=.35\textwidth
			}
			\begin{axis}[
				ylabel near ticks,
				xlabel={Dimension of truncation $d$},
				ylabel={Trotter error $b_d^{(n)}(\ket{m};t)$},
				x post scale=1,
				y post scale=1,
				legend pos=north west,
				legend cell align={left},
				legend columns=2,
				label style={font=\footnotesize},
				tick label style={font=\footnotesize},
				]
				\addplot[color=1, only marks] table[x=n, y=1, col sep=comma]{StrongX3P2_small.csv};
				\addplot[color=2, only marks] table[x=n, y=2, col sep=comma]{StrongX3P2_small.csv};
				\addplot[color=3, only marks] table[x=n, y=3, col sep=comma]{StrongX3P2_small.csv};
				\addplot[color=4, only marks] table[x=n, y=4, col sep=comma]{StrongX3P2_small.csv};
				\addplot[color=5, only marks] table[x=n, y=5, col sep=comma]{StrongX3P2_small.csv};
				\legend{\footnotesize$\ket{0}$, \footnotesize$\ket{1}$, \footnotesize$\ket{2}$, \footnotesize$\ket{3}$, \footnotesize$\ket{4}$};
			\end{axis}
		\end{tikzpicture}
		\vspace{-0.5\baselineskip}
		\caption{State-dependent error for the operators $H^{(1)}=Q^3$ and $H^{(2)}=P^2$. We consider the first five Fock-basis states $\Set{\ket{m}}=\Set{\ket{0},\dots,\ket{4}}$ for different dimensions of truncations $d=1,\dots,300$. The total evolution time is fixed to $t=1$ and the number of Trotter steps is $n=1000$. The Trotter error does not saturate, which shows that this Trotter problem does not converge.}
		\label{fig:strong_Q3_P2_samll}
		\vspace{-1\baselineskip}
	\end{figure}

	\begin{figure}
		\begin{tikzpicture}[mark size={0.5mm}, scale=1]
			\pgfplotsset{%
				width=.5\textwidth,
				height=.35\textwidth
			}
			\begin{axis}[
				ylabel near ticks,
				xlabel={Dimension of truncation $d$},
				ylabel={Trotter error $b_d^{(n)}(\ket{m};t)$},
				x post scale=1,
				y post scale=1,
				legend pos=north west,
				legend cell align={left},
				legend columns=2,
				label style={font=\footnotesize},
				tick label style={font=\footnotesize},
				]
				\addplot[color=1, only marks] table[x=n, y=1, col sep=comma]{StrongX2P2_small.csv};
				\addplot[color=2, only marks] table[x=n, y=2, col sep=comma]{StrongX2P2_small.csv};
				\addplot[color=3, only marks] table[x=n, y=3, col sep=comma]{StrongX2P2_small.csv};
				\addplot[color=4, only marks] table[x=n, y=4, col sep=comma]{StrongX2P2_small.csv};
				\addplot[color=5, only marks] table[x=n, y=5, col sep=comma]{StrongX2P2_small.csv};
				\addplot[color=1, ultra thick, domain=1:50] {sqrt(3/5)/1600};
				\addplot[color=2, ultra thick, domain=1:50] {3*sqrt(7)/8000};
				\addplot[color=3, ultra thick, domain=1:50] {sqrt(39/5)/1600};
				\addplot[color=4, ultra thick, domain=1:50] {sqrt(483)/8000};
				\addplot[color=5, ultra thick, domain=1:50] {3*sqrt(23/5)/1600};
				\legend{\footnotesize$\ket{0}$, \footnotesize$\ket{1}$, \footnotesize$\ket{2}$, \footnotesize$\ket{3}$, \footnotesize$\ket{4}$};
			\end{axis}
		\end{tikzpicture}
		\vspace{-0.5\baselineskip}
		\caption{State-dependent error for the operators $H^{(1)}=\frac{1}{2}Q^2$ and $H^{(2)}=\frac{1}{2}P^2$. We consider the first five Fock-basis states $\Set{\ket{m}}=\Set{\ket{0},\dots,\ket{4}}$ for different dimensions of truncations $d=1,\dots,50$. The total evolution time is fixed to $t=1$ and the number of Trotter steps is $n=1000$. The dots are a numerical simulation, whereas the lines show the explicit error bounds from Eq.~\eqref{Trotter_X2P2}. Since the Trotter error can be bounded independently of the truncation dimension, this Trotter problem converges.}
		\label{fig:strong_Q2_P2_small}
		\vspace{-1\baselineskip}
	\end{figure}
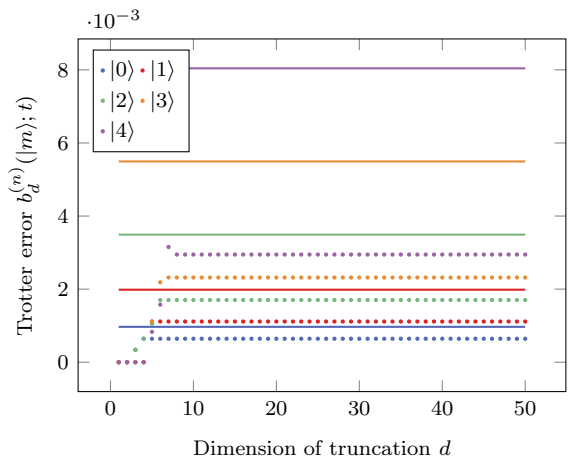

	\emph{Conclusion.}---From a practical perspective, our truncation main result reduces the complexity of determining Trotter convergence and finding infinite-dimensional error bounds to just computing finite-dimensional error bounds, which saturate in the truncation dimension. The latter may be obtained through our state-dependent error bounds from the first main result. This makes the common practice of numerical Trotter simulations rigorous.
	From a fundamental perspective, it would be great to have a generic method in order to analytically check the convergence of the infinite-dimensional Trotter product formula in a general setting. This is a hard problem and only very little literature in the mathematical physics community exists on this topic. Our result can extend to such a method in certain cases. For this, one has to be able to apply our finite-dimensional Trotter bounds to the truncation scenario. In particular, one needs that there are eigenstates of the infinite-dimensional target Hamiltonian that stay eigenstates of the truncated operators. This property does not hold in general and therefore sets a practical limitation on our method. For example, one would expect this property to break down if the infinite-dimensional target Hamiltonian does not have eigenstates. In such a case, when our first main result is not applicable to the truncation scenario but the assumptions of our second main result still hold, one might get indication for Trotter convergence by numerically studying the finite-dimensional state-dependent Trotter error. Of course, generalizing our Trotter error bounds to arbitrary input states would make this obsolete: If such a bound existed, our truncation main result would give a generic method to analytically check Trotter convergence. Both of our main results naturally generalize to Trotter products of more than two operators.

	DB thanks Paolo Facchi and Dominic Berry for interesting discussions. AH would like to thank Cahit Karg{\i} for a helpful exchange on the numerical simulations.
	LvL thanks Reinhard F.\ Werner for helpful discussions and suggestions.
	DB acknowledges funding by the Australian Research Council (project numbers FT190100106, DP210101367, CE170100009).
	NG was supported by MCIN with funding from European Union NextGenerationEU(PRTR-C17.I1) and by the Generalitat de Catalunya.
	AH was supported by the Sydney Quantum Academy.
	LvL acknowledges support by the Quantum Valley Lower Saxony.

	\bibliographystyle{prsty-title-hyperref}
	\bibliography{Finite_trotter}

\begin{thebibliography}{10}

\bibitem{Feynman1982}
R.~P. Feynman, Simulating physics with computers,
  \href{https://link.springer.com/article/10.1007/BF02650179}{Int. J. Theor.
  Phys. \textbf{21},  467  (1982)}.

\bibitem{Lloyd1996}
S. Lloyd, Universal Quantum Simulators,
  \href{https://www.science.org/doi/10.1126/science.273.5278.1073}{Science
  \textbf{273},  1073  (1996)}.

\bibitem{Poulin2015}
D. Poulin, M.~B. Hastings, D. Wecker, N. Wiebe, A.~C. Doherty, and M. Troyer,
  The Trotter step size required for accurate quantum simulation of quantum
  chemistry,
  \href{https://www.rintonpress.com/journals/qiconline.html#v15n56}{Quantum
  Inf. Comput.  361  (2015)}.

\bibitem{Childs2021}
A.~M. Childs, Y. Su, M.~C. Tran, N. Wiebe, and S. Zhu, Theory of Trotter Error
  with Commutator Scaling,
  \href{https://journals.aps.org/prx/abstract/10.1103/PhysRevX.11.011020}{Phys.
  Rev. X \textbf{11},  011020  (2021)}.

\bibitem{Klauber2013}
R.~D. Klauber, \href{http://www.quantumfieldtheory.info/}{\textit{Student
  Friendly Quantum Field Theory}} (Sandtrove Press, Fairfield (Iowa), USA,
  2013).

\bibitem{Simon2005}
B. Simon, \href{https://bookstore.ams.org/chel-351-h}{\textit{Functional
  integration and quantum physics}, 2nd ed.} (AMS Chelsea Pub, Providence
  (Rhode Island), USA, 2005).

\bibitem{Nicola2019}
F. Nicola and S.~I. Trapasso, On the Pointwise Convergence of the Integral
  Kernels in the Feynman-Trotter Formula,
  \href{https://link.springer.com/article/10.1007/s00220-019-03524-2}{Commun.
  Math. Phys. \textbf{376},  2277  (2019)}.

\bibitem{Gaveau2004}
B. Gaveau, E. Mih{\'{o}}kov{\'{a}}, M. Roncadelli, and L.~S. Schulman, Path
  integral in a magnetic field using the Trotter product formula,
  \href{https://aapt.scitation.org/doi/10.1119/1.1630334}{Am. J. Phys.
  \textbf{72},  385  (2004)}.

\bibitem{Johnson2002}
G.~W. Johnson and M.~L. Lapidus,
  \href{https://global.oup.com/academic/product/the-feynman-integral-and-feynmans-operational-calculus-9780198515722}{\textit{{T}he
  {Feynman} {I}ntegral and {F}eynman's {O}perational {C}alculus}},
  \textit{Oxford Mathematical Monographs} (Oxford University Press, New York,
  USA, 2002).

\bibitem{Misra1977}
B. Misra and E.~C.~G. Sudarshan, The Zeno's paradox in quantum theory,
  \href{https://aip.scitation.org/doi/10.1063/1.523304}{J. Math. Phys.
  \textbf{18},  756  (1977)}.

\bibitem{Arenz2018}
C. Arenz, D. Burgarth, P. Facchi, and R. Hillier, Dynamical decoupling of
  unbounded Hamiltonians,
  \href{https://aip.scitation.org/doi/full/10.1063/1.5016495}{J. Math. Phys.
  \textbf{59},  032203  (2018)}.

\bibitem{Arenz2017}
C. Arenz, D. Burgarth, and R. Hillier, Dynamical decoupling and homogenization
  of continuous variable systems,
  \href{https://iopscience.iop.org/article/10.1088/1751-8121/aa6017}{J. Phys.
  A: Math. Theor. \textbf{50},  135303  (2017)}.

\bibitem{Rivas2012}
\'{A}ngel Rivas and S.~F. Huelga,
  \href{https://link.springer.com/book/10.1007/978-3-642-23354-8}{\textit{Open
  Quantum Systems}} (Springer, Berlin, Heidelberg, Germany, 2012).

\bibitem{Fitzek2020}
F. Fitzek, J.-N. Siem{\ss}, S. Seckmeyer, H. Ahlers, E.~M. Rasel, K. Hammerer,
  and N. Gaaloul, Universal atom interferometer simulation of elastic
  scattering processes,
  \href{https://www.nature.com/articles/s41598-020-78859-1}{Nat. Sci. Rep.
  \textbf{10},  22120  (2020)}.

\bibitem{Sieberer2019}
L.~M. Sieberer, T. Olsacher, A. Elben, M. Heyl, P. Hauke, F. Haake, and P.
  Zoller, Digital quantum simulation, Trotter errors, and quantum chaos of the
  kicked top, \href{https://doi.org/10.1038/s41534-019-0192-5}{npj Quantum Inf.
  \textbf{5},    (2019)}.

\bibitem{Chinni2022}
K. Chinni, M.~H. Mu{\~{n}}oz-Arias, I.~H. Deutsch, and P.~M. Poggi, Trotter
  Errors from Dynamical Structural Instabilities of Floquet Maps in Quantum
  Simulation, \href{https://doi.org/10.1103/PRXQuantum.3.010351}{{PRX} Quantum
  \textbf{3},  010351  (2022)}.

\bibitem{Kargi2021}
C. Kargi, J.~P. Dehollain, F. Henriques, L.~M. Sieberer, T. Olsacher, P. Hauke,
  M. Heyl, P. Zoller, and N.~K. Langford, Quantum Chaos and Universal
  Trotterisation Behaviours in Digital Quantum Simulations,
  \href{https://doi.org/10.48550/arXiv.2110.11113}{  (2021)}.

\bibitem{Verstraete2004}
F. Verstraete, J.~J. Garc{\'{\i}}a-Ripoll, and J.~I. Cirac, Matrix Product
  Density Operators: Simulation of Finite-Temperature and Dissipative Systems,
  \href{https://journals.aps.org/prl/abstract/10.1103/PhysRevLett.93.207204}{Phys.
  Rev. Lett. \textbf{93},  207204  (2004)}.

\bibitem{Zwolak2004}
M. Zwolak and G. Vidal, Mixed-State Dynamics in One-Dimensional Quantum Lattice
  Systems: A Time-Dependent Superoperator Renormalization Algorithm,
  \href{https://journals.aps.org/prl/abstract/10.1103/PhysRevLett.93.207205}{Phys.
  Rev. Lett. \textbf{93},  207205  (2004)}.

\bibitem{Vidal2004}
G. Vidal, Efficient Simulation of One-Dimensional Quantum Many-Body Systems,
  \href{https://journals.aps.org/prl/abstract/10.1103/PhysRevLett.93.040502}{Phys.
  Rev. Lett. \textbf{93},  040502  (2004)}.

\bibitem{Cincio2008}
L. Cincio, J. Dziarmaga, and M.~M. Rams, Multiscale Entanglement
  Renormalization Ansatz in Two Dimensions: Quantum Ising Model,
  \href{https://journals.aps.org/prl/abstract/10.1103/PhysRevLett.100.240603}{Phys.
  Rev. Lett. \textbf{100},  240603  (2008)}.

\bibitem{Lambert2019}
N. Lambert, S. Ahmed, M. Cirio, and F. Nori, Modelling the ultra-strongly
  coupled spin-boson model with unphysical modes,
  \href{https://www.nature.com/articles/s41467-019-11656-1}{Nat. Comm.
  \textbf{10},    (2019)}.

\bibitem{Jeckelmann1998}
E. Jeckelmann and S.~R. White, Density-matrix renormalization-group study of
  the polaron problem in the Holstein model,
  \href{https://journals.aps.org/prb/abstract/10.1103/PhysRevB.57.6376}{Physical
  Review B \textbf{57},  6376  (1998)}.

\bibitem{Vidal2007}
G. Vidal, Classical Simulation of Infinite-Size Quantum Lattice Systems in One
  Spatial Dimension, \href{https://doi.org/10.1103/PhysRevLett.98.070201}{Phys.
  Rev. Lett. \textbf{98},  070201  (2007)}.

\bibitem{Suzuki1985}
M. Suzuki, Decomposition formulas of exponential operators and Lie exponentials
  with some applications to quantum mechanics and statistical physics,
  \href{http://aip.scitation.org/doi/10.1063/1.526596}{J. Math. Phys.
  \textbf{26},  601  (1985)}.

\bibitem{Kato1978}
T. Kato,  in \href{}{\textit{Topics in Functional Analysis. Ad. Math. Suppl.
  Studies}}, edited by I. Gohberg and M. Kac (Academic Press, New York, USA,
  1978), Vol.~3, pp.\ 185--195.

\bibitem{Lapidus1981}
M.~L. Lapidus, Generalization of the Trotter-Lie formula,
  \href{https://link.springer.com/article/10.1007/BF01697972}{Integral Equ.
  Oper. Theory \textbf{4},  366  (1981)}.

\bibitem{Rabi1936}
I.~I. Rabi, On the Process of Space Quantization,
  \href{https://doi.org/10.1103/PhysRev.49.324}{Phys. Rev. \textbf{49},  324
  (1936)}.

\bibitem{Rabi1937}
I.~I. Rabi, Space Quantization in a Gyrating Magnetic Field,
  \href{https://doi.org/10.1103/PhysRev.51.652}{Phys. Rev. \textbf{51},  652
  (1937)}.

\bibitem{Jaynes1963}
E.~T. Jaynes and F.~W. Cummings, Comparison of quantum and semiclassical
  radiation theories with application to the beam maser,
  \href{https://ieeexplore.ieee.org/document/1443594}{Proc. {IEEE} \textbf{51},
   89  (1963)}.

\bibitem{Gottesman2001}
D. Gottesman, A. Kitaev, and J. Preskill, Encoding a qubit in an oscillator,
  \href{https://doi.org/10.1103/PhysRevA.64.012310}{Phys. Rev. A \textbf{64},
  012310  (2001)}.

\bibitem{Koch2007}
J. Koch, T.~M. Yu, J. Gambetta, A.~A. Houck, D.~I. Schuster, J. Majer, A.
  Blais, M.~H. Devoret, S.~M. Girvin, and R.~J. Schoelkopf, Charge-insensitive
  qubit design derived from the Cooper pair box,
  \href{https://doi.org/10.1103/PhysRevA.76.042319}{Phys. Rev. A \textbf{76},
  042319  (2007)}.

\bibitem{Schreier2008}
J.~A. Schreier, A.~A. Houck, J. Koch, D.~I. Schuster, B.~R. Johnson, J.~M.
  Chow, J.~M. Gambetta, J. Majer, L. Frunzio, M.~H. Devoret, S.~M. Girvin, and
  R.~J. Schoelkopf, Suppressing charge noise decoherence in superconducting
  charge qubits, \href{https://doi.org/10.1103/PhysRevB.77.180502}{Phys. Rev. B
  \textbf{77},  180502  (2008)}.

\bibitem{Arute2019}
F. Arute, K. Arya, R. Babbush, D. Bacon, J.~C. Bardin, R. Barends, R. Biswas,
  S. Boixo, F.~G. S.~L. Brandao, D.~A. Buell, B. Burkett, Y. Chen, Z. Chen, B.
  Chiaro, R. Collins, W. Courtney, A. Dunsworth, E. Farhi, B. Foxen, A. Fowler,
  C. Gidney, M. Giustina, R. Graff, K. Guerin, S. Habegger, M.~P. Harrigan,
  M.~J. Hartmann, A. Ho, M. Hoffmann, T. Huang, T.~S. Humble, S.~V. Isakov, E.
  Jeffrey, Z. Jiang, D. Kafri, K. Kechedzhi, J. Kelly, P.~V. Klimov, S. Knysh,
  A. Korotkov, F. Kostritsa, D. Landhuis, M. Lindmark, E. Lucero, D. Lyakh, S.
  Mandr{\`{a}}, J.~R. McClean, M. McEwen, A. Megrant, X. Mi, K. Michielsen, M.
  Mohseni, J. Mutus, O. Naaman, M. Neeley, C. Neill, M.~Y. Niu, E. Ostby, A.
  Petukhov, J.~C. Platt, C. Quintana, E.~G. Rieffel, P. Roushan, N.~C. Rubin,
  D. Sank, K.~J. Satzinger, V. Smelyanskiy, K.~J. Sung, M.~D. Trevithick, A.
  Vainsencher, B. Villalonga, T. White, Z.~J. Yao, P. Yeh, A. Zalcman, H.
  Neven, and J.~M. Martinis, Quantum supremacy using a programmable
  superconducting processor,
  \href{https://doi.org/10.1038/s41586-019-1666-5}{Nature \textbf{574},  505
  (2019)}.

\bibitem{Sahinoglu2021}
B. {\c{S}}ahino{\u{g}}lu and R.~D. Somma, Hamiltonian simulation in the
  low-energy subspace,
  \href{http://www.nature.com/articles/s41534-021-00451-w}{npj Quantum Inf.
  \textbf{7},    (2021)}.

\bibitem{Yi2022}
C. Yi and E. Crosson, Spectral analysis of product formulas for quantum
  simulation, \href{https://www.nature.com/articles/s41534-022-00548-w}{npj
  Quantum Inf. \textbf{8},    (2022)}.

\bibitem{Becker2021}
S. Becker, N. Datta, L. Lami, and C. Rouz{\'{e}}, Energy-Constrained
  Discrimination of Unitaries, Quantum Speed Limits, and a Gaussian
  Solovay-Kitaev Theorem,
  \href{https://link.aps.org/doi/10.1103/PhysRevLett.126.190504}{Phys. Rev.
  Lett. \textbf{126},  190504  (2021)}.

\bibitem{Luijk2022}
L. van Luijk, N. Galke, A. Hahn, and D. Burgarth, Error bounds for Lie group
  representations in quantum mechanics,
  \href{https://iopscience.iop.org/article/10.1088/1751-8121/ad288b}{J. Phys.
  A: Math. Theor.  (2024)}.

\bibitem{Jahnke2000}
T. Jahnke and C. Lubich, Error Bounds for Exponential Operator Splittings,
  \href{https://link.springer.com/article/10.1023/A:1022396519656}{BIT
  \textbf{40},  735  (2000)}.

\bibitem{An2021}
D. An, D. Fang, and L. Lin, Time-dependent unbounded Hamiltonian simulation
  with vector norm scaling,
  \href{https://doi.org/10.22331/q-2021-05-26-459}{Quantum \textbf{5},  459
  (2021)}.

\bibitem{Ichinose2004}
T. Ichinose and H. Tamura, Note on the Norm Convergence of the Unitary Trotter
  Product Formula, \href{https://doi.org/10.1007/s11005-004-3760-2}{Lett. Math.
  Phys. \textbf{70},  65  (2004)}.

\bibitem{ReedSimon1981}
M. Reed and B. Simon,
  \href{https://www.ebook.de/de/product/3681260/michael_duke_university_north_carolina_reed_barry_princeton_university_new_jersey_simon_i_functional_analysis.html}{\textit{I:
  Functional Analysis}, revisited and enlarged ed.} (Academic Press Inc., San
  Diego (California), USA, 1981).

\bibitem{Burgarth2022}
D. Burgarth, P. Facchi, G. Gramegna, and K. Yuasa, One bound to rule them all:
  from Adiabatic to Zeno,
  \href{https://quantum-journal.org/papers/q-2022-06-14-737/}{Quantum
  \textbf{6},  737  (2022)}.

\bibitem{EN01}
K.-J. Engel and R. Nagel, One-parameter semigroups for linear evolution
  equations,
  \href{https://link.springer.com/article/10.1007/s002330010042}{Semigroup
  Forum \textbf{63},  278  (2001)}.

\bibitem{Duffield1992}
N.~G. Duffield and R.~F. Werner, Mean-field dynamical semigroups on
  $C^*$-algebras,
  \href{https://www.worldscientific.com/doi/abs/10.1142/S0129055X92000108}{Rev.
  Math. Phys. \textbf{04},  383  (1992)}.

\bibitem{Zhu1993}
C. Zhu and J.~R. Klauder, Classical symptoms of quantum illnesses,
  \href{https://doi.org/10.1119/1.17221}{Am. J. Phys. \textbf{61},  605
  (1993)}.

\end{thebibliography}

	\newpage
	
	\onecolumngrid

	\begin{center}
		{\Large\textbf{Supplementary Material}\par}
	\end{center}
	
		\hypertarget{app:preliminaries}{}
	\section{A. Mathematical Preliminaries}
	
	In this section, we introduce all the mathematical concepts and the notation needed for the proof of our main results. We also elaborate more on the mathematical setting and motivate why we have to consider state-dependent error bounds instead of norm error bounds.
	
	Consider a separable infinite-dimensional Hilbert space $\H$ with an orthonormal basis $\Set{\ket{j}}$.
	A corresponding sequence of truncated finite-dimensional Hilbert spaces of increasing dimension $d<\infty$ is then given by $V_d = \mathrm{span}\Set{\ket 0,\dots,\ket{d-1}}\subset\H$.
	Denote the orthogonal projection onto $V_d$ by $P_d=\sum_{j=0}^{d-1}\ket{j}\bra{j}$.
	We denote the union of the truncated Hilbert spaces by $\V:=\bigcup_d V_d$. Since $\{\ket j \}$ is a basis of $\H$, $\V$ lies dense in $\H$, i.e.\ $\V^\perp=\Set{0}$.
	The scalar product of $\H$ is denoted by $\langle\placeholder|\placeholder\rangle$ and its induced norm by $\Vert\psi\Vert=\sqrt{\langle\psi|\psi\rangle}$, $\psi\in\H$.
	
	A \emph{strongly continuous contraction semigroup} is a set of linear operators $T(t)$ on $\H$, $t\ge 0$, with the following properties
	\begin{enumerate}
		\item semigroup: for all $t,s\in\RR_{\geq 0}$, $T(t)T(s)=T(t+s)$ and $T(0)=\id$,
		\item contraction: $\|T(t)\ket{\psi}\|\leq\|\ket{\psi}\|$, for all $\ket{\psi}\in\H$,
		\item strong continuity: For all $\ket{\psi}\in\H$, $\lim_{t\downarrow 0}\|T(t)\ket{\psi}-\ket{\psi}\|=0$.
	\end{enumerate}
	Every strongly continuous contraction semigroup defines an unbounded operator $K:\dom K\to\H$ by setting $\dom K$ to be the set such that $K\psi := \lim_{t\downarrow 0} \tfrac \rmi t (T(t)\psi - \psi)$ exists \cite{EN01}.
	As is typical for unbounded operators $\dom K$ will usually be a proper subset of $\H$.
	In fact, $\dom K = \H$ if and only if the semigroup is uniformly continuous.
	Here, \emph{uniform} refers to the topology induced by the operator norm $\|A\|_\infty=\sup_{\|\ket{\psi}\|=1}\|A\ket{\psi}\|$.
	In finite dimensions, the uniform and the strong operator topology coincide, but in infinite dimensions this is wrong and the uniform topology is stronger.
	The operator $K$ is then called the \emph{generator} of the semigroup and (in principle) encodes all information necessary to compute $T(t)$.
	Self-adjoint (unbounded) operators $H$ on $\H$ are precisely the generators of unitary groups $U(t) = \rme^{-\rmi tH}$.
	We will prove a version of Main~Result~\ref{thm:mainthm} from the main text for strongly continuous contraction semigroups.
	The unitary case as considered in the main text then follows from that.
	See Thm.~\ref{thm:main} in Sec.~\hyperlink{app:proof}{C} of this Supplementary Material.
	
	Next, we recall the Trotter product formula.
	This discussion will motivate the use of state-dependent error bounds.
	For this purpose, let us focus on the finite-dimensional unitary case first: Let $\dim(\H)<\infty$ and let $H^{(1)}:\H\to\H$, $H^{(2)}:\H\to\H$ be self-adjoint, hence generators of two continuous unitary groups. Then the \emph{uniform Trotter error} satisfies the bound \cite{Childs2021,Suzuki1985}
	\begin{align}
		\Big\|\left(\rme^{-\rmi \frac{t}{n}H^{(1)}}\rme^{-\rmi \frac{t}{n}H^{(2)}}\right)^n-\rme^{-\rmi t(H^{(1)}+H^{(2)})}\Big\|_\infty
		\leq
		\frac{t^2}{2n}\big\|[H^{(1)},H^{(2)}]\big\|_\infty.\label{eq:uniform_error_bound}
	\end{align}
	Due to the explicit bound with $\mathcal{O}(1/N)$ scaling, we can conclude from Eq.~\eqref{eq:uniform_error_bound} that the finite-dimensional Trotter product formula \emph{always} converges uniformly, i.e.\ the uniform Trotter error goes to zero as $n\to\infty$.
	But what does Eq.~\eqref{eq:uniform_error_bound} imply in the situation where the Hamiltonians are finite-dimensional truncations of infinite-dimensional operators?
	Let $H^{(1)}:\D(H^{(1)})\to\H$ and $H^{(2)}:\D(H^{(2)})\to\H$ be two self-adjoint operators acting on a separable infinite-dimensional Hilbert space $\H$. We do not make any assumptions on the spectra of these two operators, in particular we allow for continuous spectra.
	Denote their finite-dimensional approximations as $H_d^{(i)} := P_d H^{(i)}P_d : V_d\to V_d$ with $i=1,2$.
	Obviously, Eq.~\eqref{eq:uniform_error_bound} holds for $H^{(1)}_n$ and $H^{(2)}_n$, so that
	\begin{align}
		\beta_d^{(n)}(t):=\Big\|\left(\rme^{-\rmi \frac{t}{n}H_d^{(1)}}\rme^{-\rmi \frac{t}{n}H_d^{(2)}}\right)^n-\rme^{-\rmi t(H_d^{(1)}+H_d^{(2)})}\Big\|_\infty
		\leq
		\frac{t^2}{2n}\big\|[H_d^{(1)},H_d^{(2)}]\big\|_\infty.\label{eq:uniform_error_bound_beta}
	\end{align}
	In this case, for a fixed total evolution time $t$ and a fixed number of Trotter steps $n$, the commutator error $\|[H_d^{(1)},H_d^{(2)}]\|_\infty$ will typically diverge with the dimension $d$ of truncation.
	This is due to the fact that the operator norm $\Vert\placeholder\Vert_\infty$ gives a worst-case error by taking the supremum over all normalised $\ket{\psi}\in V_d$.
	By increasing $d$, new vectors will be added so that the Trotter error in Eq.~\eqref{eq:uniform_error_bound_beta} becomes larger.
	Consequently, the Trotter approximants might fail to converge uniformly.
	In fact:
	\begin{prop}\label{thm:uniformTrotter}
		If the uniform Trotter error $\beta_d^{(n)}(t)$ is bounded independently of $d$ then the infinite-dimensional Trotter product converges uniformly.
		\begin{proof}
			A proof will be given in the next section.
		\end{proof}
	\end{prop}
	This will, of course, not hold for arbitrary Hamiltonians.
	For an example, see \cite[Sec.~3]{Ichinose2004}.
	Indeed, this can also be observed numerically:
	As an example, we study the operators $H^{(1)}=\frac{1}{2}\left(Q^2+P^2\right)$ and $H^{(2)}=\frac{1}{2}\left(QP+PQ\right)$ with $Q$ denoting the position operator and $P$ the momentum operator.
	Note, that $H^{(1)}$ is the Hamiltonian of the quantum harmonic oscillator and that $H^{(2)}$ generates squeezing transformations.
	As can be seen from Fig.~\ref{fig:Uniform}, the Trotter error $\beta_d^{(n)}(t)$ increases with the truncation dimension $d$ for this example and eventually reaches $2$, which is the maximum of the operator norm distance of two unitaries.
	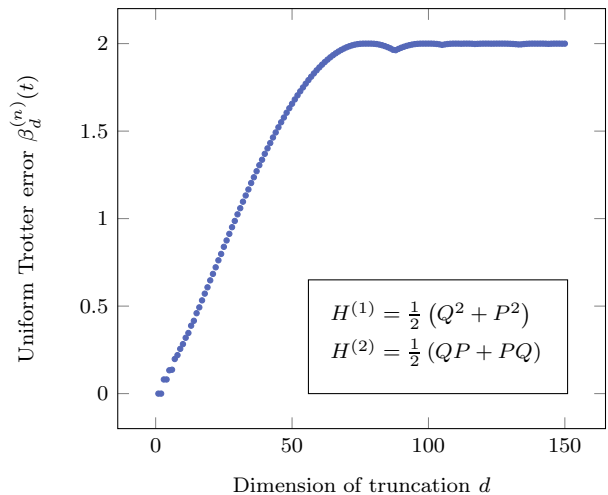
\begin{figure}
		\begin{tikzpicture}[mark size={1}]
			\pgfplotsset{%
				width=.45\textwidth,
				height=.4\textwidth
			}
			\begin{axis}[
				xlabel={Dimension of truncation $d$},
				ylabel={Uniform Trotter error $\beta_d^{(n)}(t)$},
				x post scale=1,
				y post scale=1,
				legend style={cells={align=left}},
				legend columns = 1,
				legend pos = south east,
				label style={font=\footnotesize},
				tick label style={font=\footnotesize},
				]
				\addplot[color=1, only marks] table[x=n, y=norm, col sep=comma]{Uniform.csv};
				\draw (55,0) rectangle (150,65);
				\node at (100,45) (A1) {\footnotesize$H^{(1)}=\frac{1}{2}\left(Q^2+P^2\right)$};
				\node[below=0.5cm of A1.west, anchor=west] (A2) {\footnotesize$H^{(2)}=\frac{1}{2}\left(QP+PQ\right)$};
			\end{axis}
		\end{tikzpicture}
		\caption{
			Uniform Trotter error for $H^{(1)}=\frac{1}{2}\left(Q^2+P^2\right)$ and $H^{(2)}=\frac{1}{2}\left(QP+PQ\right)$ with $Q$ the position and $P$ the momentum operator. Parameters: $n=20$ Trotter steps, total evolution time of $t=2$. The error reaches the maximal norm distance of $2$ showing that the Trotter formula does not converge in norm.
		}
		\label{fig:Uniform}
	\end{figure}
	\noindent
	For this reason, we cannot expect uniform convergence in the infinite-dimensional setting.
	Rather, we need a fine-tuned error measure.
	This can be achieved by considering the state-dependent, or \emph{strong}, Trotter error
	\begin{equation}
		b_d^{(n)}(\ket{\psi}; t):=\Big\|\left(\rme^{\rmi \frac{t}{n}H_d^{(1)}}\rme^{\rmi \frac{t}{n}H_d^{(2)}}\right)^n\ket{\psi}-\rme^{\rmi t(H_d^{(1)}+H_d^{(2)})}\ket{\psi}\Big\|\label{eq:strong_Trotter_error_appendix}
	\end{equation}
	instead.
	If $b_d^{(n)}(\ket{\psi}; t)\to 0$ for $n\to\infty$ for all $\ket\psi$, we say that the Trotter product formula \emph{converges strongly}.
	See \cite[Sec.~VI.1]{ReedSimon1981} for details regarding strong convergence.
	In the main text, we used Eq.~\eqref{eq:strong_Trotter_error_appendix} to define
	\begin{equation}
		b^{(n)}(\ket{\psi} ; t) := \limsup_{d\to\infty}  b_d^{(n)}(\ket{\psi};t).\label{eq:Trotter_error_infinite}
	\end{equation}
	
	Lastly, consider a linear operator $H:\D(H)\to\H$ acting on an infinite-dimensional separable Hilbert space $\H$.
	Then the graph norm $\|\placeholder\|_H$ of $H$ is defined for $\ket{\psi}\in\D(H)$ by
	\begin{equation}
		\|\ket{\psi}\|_H^2=\|\ket{\psi}\|^2+\|H\ket{\psi}\|^2.\label{eq:GrapH_dorm}
	\end{equation}
	$H$ is called \emph{closed} if $\dom H$ is complete with respect to the graph norm $\norm\placeholder_H$.
	If $H_1$ is another operator, then we write $H\subset H_1$ and call $H_1$ an \emph{extension} of $H$ if $\dom H \subset\dom{H_1}$ and if $H_1\ket\psi = H\ket\psi$ for $\ket\psi\in\dom H$.
	$H$ is called \emph{densely defined} if $\dom H$ is dense in $\H$.
	In that case, its adjoint $H^*$ is a well-defined closed operator.
	If $H$ is \emph{symmetric}, i.e. $\ip\psi{H\varphi} = \ip{H\psi}\varphi$, then $H^{**}$ is the smallest closed extension of $H$, called the \emph{closure}.
	
	We say that a subspace $\W\subset\dom H$ is a \emph{core} for $H$ if $\W$ is dense in $\D(H)$ with respect to the graph norm of $H$.

	\hypertarget{app:proof_bound}{}
	\section{B. Proof of the first main result}
	
	This section demonstrates the proof of our Main~Result~\ref{thm:finite_bounds} from the main text. Throughout, we will assume that $\mathcal{{H}}$ is a finite-dimensional Hilbert space. We then consider two Hamiltonians on $\H$, i.e.\ self-adjoint operators $H^{(1)}:\H\to\H$ and $H^{(2)}:\H\to\H$, whose Trotter product is examinated. Denote
	\begin{equation}
		U(t)=\rme^{-\rmi(H^{(1)}+H^{(2)})t}\label{eq:W(t)}
	\end{equation}
	and
	\begin{equation}
		W^{(n)}(t)=\left(\rme^{-\rmi H^{(1)}\frac{t}{n}}\rme^{-\rmi H^{(2)}\frac{t}{n}}\right)^{n}.\label{eq:WN(t)}
	\end{equation}
	We now recapitulate \cite[Lemma~1]{Burgarth2022}, in particular Eq.~(2.4) therein. This will be needed in order to prove our Main~Result~\ref{thm:finite_bounds}.
	
	\begin{lem}\label{lemma:action}
		Let $\tilde{H}_1(t)$, $\tilde{H}_2(t)$ be two families of time-dependent, self-adjoint and locally integrable operators. Define the corresponding unitaries they generate
		\begin{equation}
			\tilde{U}_j(t)=\T\exp\left(-\rmi\int_{0}^{t} \tilde{H}_j(s)\mathrm{d}s\right),\quad j=1,2,
		\end{equation}
		where $\T$ denotes time-ordering. Furthermore, define the integral action $S_{21}(t)$ as
		\begin{equation}
			S_{21}(t)=\int_{0}^{t}\left[\tilde{H}_2(s)-\tilde{H}_1(s)\right]\mathrm{d}s.
		\end{equation}
		Then for all $t\in\RR_{\geq0}$,
		\begin{equation}
			\tilde{U}_2(t)-\tilde{U}_1(t)=-\rmi S_{21}(t)\tilde{U}_2(t)-\int_{0}^{t} \tilde{U}_1(t)\tilde{U}_1(s)^\dagger\left[\tilde{H}_1(s)S_{21}(s)-S_{21}(s)\tilde{H}_2(s)\right]\tilde{U}_2(s)\mathrm{d}s.
		\end{equation}
		\begin{proof}
			This is proved in \cite[Lemma~1]{Burgarth2022}.
		\end{proof}
	\end{lem}
	
	Let us restate our Main~Result~\ref{thm:finite_bounds} in a slightly generalized version.
	
	\begin{main_one}
		Let $\H$ be a finite-dimensional Hilbert space and let  $H^{(1)}:\H\to\H$ and $H^{(2)}:\H\to\H$ be self-adjoint operators on $\H$. Define $U(t)$ and $W^{(n)}(t)$ as in Eq.~\eqref{eq:W(t)} and Eq.~\eqref{eq:WN(t)}, respectively. If $\ket{\varphi}$ is an eigenvector of the self-adjoint operator $H^{(1)}+H^{(2)}$ accoring to the eigenvalue equation $(H^{(1)}+H^{(2)})\ket{\varphi}=h\ket{\varphi}$, then we have for all $t\in\RR_{\geq0}$ and $n\in\NN$,
		\begin{equation}
			\left\Vert \left(U(t)-W^{(n)}(t)\right)\ket{\varphi}\right\Vert \leq \frac{t^2}{2n}\inf_{\alpha\in\RR}\left(\Big\Vert \Big[H^{(1)}-\alpha h\Big]^2\ket{\varphi}\Big\Vert + \Big\Vert \Big[H^{(2)}-(1-\alpha) h\Big]^2\ket{\varphi}\Big\Vert\right).\label{eq:second_main_generalized}
		\end{equation}
		
		\begin{proof}
			We start by applying Lemma~\ref{lemma:action}. For this, take $\tilde{U}_{1}(2t)=W^{(n)}(t)$ and $\tilde{U}_2(2t)=U(t)$. Notice that $\tilde{U}_1$ is generated by a piece-wise constant Hamiltonian
			$\tilde{H}_{1}(s)$ over the time interval $[0,2t].$ This Hamiltonian is
			defined by
			\begin{equation}
				\tilde{H}_{1}(s)\equiv\begin{cases}
					H^{(2)}, \quad & t\in\left[0,\frac{t}{n}\right)\\
					H^{(1)}, \quad & t\in\left[\frac{t}{n},\frac{2t}{n}\right)
				\end{cases}
			\end{equation}
			and then periodically extended by $\tilde{H}_{1}(s+\frac{2t}{n})=\tilde{H}_{1}(s)$. The unitary $\tilde{U}_2$ is simply generated by the average constant Hamiltonian $\tilde{H}_{2}(s)=\frac{H^{(1)}+H^{(2)}}{2}$. Since
			both Hamiltonians are locally integrable, Lemma~\ref{lemma:action} can be applied. We
			want to estimate the error on an eigenvector $\left(H^{(1)}+H^{(2)}\right)\ket{\varphi}=h\ket{\varphi}:$
			\begin{equation}
				\left(U(t)-W^{(n)}(t)\right)\ket{\varphi}=-\rmi\rme^{-\rmi ht}S_{21}(2t)\ket{\varphi}-\int_{0}^{2t}\rme^{-\rmi \frac{h}{2}s}\tilde{U}_{1}(2t)\tilde{U}_{1}^{\dagger}(s)\left[\tilde{H}_{1}(s)-\frac{h}{2}\right]S_{21}(s)\ket{\varphi}\mathrm{d}s.
			\end{equation}
			The action reads
			\begin{align}
				S_{21}(s) & =\int_{0}^{s}\left(\frac{H^{(1)}+H^{(2)}}{2}-\tilde{H}_{1}(\tau)\right)\mathrm{d}\tau\nonumber\\
				& =\int_{0}^{\left\lfloor s/\frac{2t}{n}\right\rfloor \frac{2t}{n}}\left(\frac{H^{(1)}+H^{(2)}}{2}-\tilde{H}_{1}(\tau)\right)\mathrm{d}\tau+\int_{\left\lfloor s/\frac{2t}{n}\right\rfloor \frac{2t}{n}}^{s}\left(\frac{H^{(1)}+H^{(2)}}{2}-\tilde{H}_{1}(\tau)\right)\mathrm{d}\tau\nonumber\\
				& =\int_{\left\lfloor s/\frac{2t}{n}\right\rfloor \frac{2t}{n}}^{s}\left(\frac{H^{(1)}+H^{(2)}}{2}-\tilde{H}_{1}(\tau)\right)\mathrm{d}\tau\nonumber\\
				& =\int_{0}^{\left\{ s/\frac{2t}{n}\right\} }\left(\frac{H^{(1)}+H^{(2)}}{2}-\tilde{H}_{1}(\tau)\right)\mathrm{d}\tau,
			\end{align}
			where we split the integral into full periods and a remainder and
			used the periodicity of $\tilde{H}_{1}(s).$ In particular, $S_{21}(2t)=0$
			and we can bound
			\begin{equation}
				\left\Vert \left(U(t)-W^{(n)}(t)\right)\ket{\varphi}\right\Vert  \le\int_{0}^{2t}\left\Vert \left[\tilde{H}_{1}(s)-\frac{h}{2}\right]S_{21}(s)\ket{\varphi}\right\Vert\mathrm{d}s.
			\end{equation}
			For $s\in\left[0,\frac{t}{n}\right)$, we have
			\begin{equation}
				S_{21}(s)=\int_{0}^{s}\left(\frac{H^{(1)}+H^{(2)}}{2}-H^{(2)}\right)\mathrm{d}\tau=s\frac{H^{(1)}-H^{(2)}}{2}.
			\end{equation}
			Analogously for $s\in\left[\frac{t}{n},\frac{2t}{n}\right)$, we obtain
			\begin{equation}
				S_{21}(s)=\frac{t}{n}\frac{H^{(1)}-H^{(2)}}{2}+\int_{\frac{t}{n}}^{s}\left(\frac{H^{(1)}+H^{(2)}}{2}-H^{(1)}\right)\mathrm{d}\tau=\left(\frac{2t}{n}-s\right)\frac{H^{(1)}-H^{(2)}}{2}.
			\end{equation}
			Therefore,
			\begin{align}
				\left\Vert \left(U(t)-W^{(n)}(t)\right)\ket{\varphi}\right\Vert 
				\leq&
				n\int_0^{\frac{t}{n}} \frac{s}{2}\left\Vert \left[H^{(2)}-\frac{h}{2}\right]\left(H^{(1)}-H^{(2)}\right)\ket{\varphi}\right\Vert\mathrm{d}s\nonumber\\
				&+n\int_{\frac{t}{n}}^{\frac{2t}{n}} \left(\frac{t}{n}-\frac{s}{2}\right)\left\Vert \left[H^{(1)}-\frac{h}{2}\right]\left(H^{(1)}-H^{(2)}\right)\ket{\varphi}\right\Vert\mathrm{d}s\nonumber\\
				=&\frac{t^2}{4n}\left(\left\Vert \left[H^{(2)}-\frac{h}{2}\right]\left(H^{(1)}-H^{(2)}\right)\ket{\varphi}\right\Vert+\left\Vert \left[H^{(1)}-\frac{h}{2}\right]\left(H^{(1)}-H^{(2)}\right)\ket{\varphi}\right\Vert\right),
			\end{align}
			where we used the periodicity of $\tilde{H}_{1}(s)$ and of $S_{21}(s)$ in the first step.
			We expand
			\begin{align}
				\left\Vert\left[H^{(2)}-\frac{h}{2}\right]\left(H^{(1)}-H^{(2)}\right)\ket{\varphi}\right\Vert
				& = \left\Vert\left[H^{(2)}-\frac{h}{2}\right]\left(H^{(1)}+H^{(2)}-2H^{(2)}\right)\ket{\varphi}\right\Vert\nonumber\\
				&=\left\Vert\left[H^{(2)}-\frac{h}{2}\right]\left(h-2H^{(2)}\right)\ket{\varphi}\right\Vert\nonumber\\
				&=2\left\Vert\left[H^{(2)}-\frac{h}{2}\right]^2\ket{\varphi}\right\Vert
			\end{align}
			and simliar for the second term
			\begin{equation}
				\left\Vert\left[H^{(1)}-\frac{h}{2}\right]\left(H^{(1)}-H^{(2)}\right)\ket{\varphi}\right\Vert
				=2\left\Vert\left[H^{(1)}-\frac{h}{2}\right]^2\ket{\varphi}\right\Vert.
			\end{equation}
			Therefore, in total we receive
			\begin{equation}
				\left\Vert \left(U(t)-W^{(n)}(t)\right)\ket{\varphi}\right\Vert \leq \frac{t^2}{2n}\left(\Big\Vert \Big[H^{(1)}-\frac{h}{2}\Big]^2\ket{\varphi}\Big\Vert + \Big\Vert \Big[H^{(2)}-\frac{h}{2}\Big]^2\ket{\varphi}\Big\Vert\right),\label{eq:second_main_half}
			\end{equation}
			which is the bound from our Main~Result~\ref{thm:finite_bounds} in the main text. In order to show the slightly generalized version of this bound stated in Eq.~\eqref{eq:second_main_generalized}, notice that we can always rescale the Hamiltonians $H^{(1)}$ and $H^{(2)}$ by a multiple of the identity $\1$ without changing the Trotter error. That is, by rescaling $H^{(1)}\rightarrow H^{(1)}+a\1$ and $H^{(2)}\rightarrow H^{(2)}+b\1$ we will not change $\Vert \left(U(t)-W^{(n)}(t)\right)\ket{\varphi}\Vert$ since $\1$ commutes with any Hamiltonian. This procedure leads to a rescaling $h\rightarrow h+a+b$ and Eq.~\eqref{eq:second_main_half} becomes
			\begin{equation}
				\left\Vert \left(U(t)-W^{(n)}(t)\right)\ket{\varphi}\right\Vert \leq \frac{t^2}{2n}\left(\Big\Vert \Big[H^{(1)}-\frac{h+b-a}{2}\Big]^2\ket{\varphi}\Big\Vert + \Big\Vert \Big[H^{(2)}-\frac{h+a-b}{2}\Big]^2\ket{\varphi}\Big\Vert\right),
			\end{equation}
			which is true for any $a,b\in\RR$. The assertion then follows by defining $\alpha=(h+b-a)/(2h)$ and taking the infimum over $\alpha$. The latter is possible since the inequality holds for all $\alpha\in\RR$.
		\end{proof}
	\end{main_one}

	\hypertarget{app:proof}{}
	\section{C. Proof of the second main result}
	
	In this section, we present the proof of Main~Result~\ref{thm:mainthm} from the main text. For this, we will use the notation introduced in Sec.~\hyperlink{app:preliminaries}{A} of this Supplementary Material. We will first prove a more general result, which holds for strongly continuous contraction semigroups. This is given by Thm.~\ref{thm:main}. Our Main~Result~\ref{thm:mainthm}
	will then follow as a special case. In the proofs of the theorems, we will use the following consequence of the \emph{Trotter-Kato approximation theorem} \cite[Sec.~4]{EN01}:
	
	\begin{thm}\label{thm:TKAT}
		For each $n\in\NN$, let $T_d(t)$ be a contraction semigroup on $V_d$.
		Let $H_d$ be the sequence of generators.
		The following statements are equivalent
		\begin{enumerate}[(a)]
			\item \label{it:TKgroup}
			The strong limit $T(t) = \slim_{d\to\infty} T_d(t)$ exists for all times $t>0$ and defines a strongly continuous contraction semigroup.
			\item \label{it:TKgen}
			The operator $H : \dom H \to \H$, defined by
			\begin{align*}
				\ket{\psi} \in \D(H) \iff &\text{there exists a sequence } (\ket{\psi_d})_{d\in\NN} \in \CCC, \text{ such that } \ket{\psi_d} \to \ket{\psi}\in\H, \\
				&\text{and } H\ket{\psi} :=  \lim_{d\to\infty} H_d\ket{\psi_d} \text{ exists},
			\end{align*}
			is actually well-defined in the sense that $H\ket{\psi}$ does not depend on the chosen sequence $(\ket{\psi_d})$. Additionally, $H$ generates a contraction semigroup.
		\end{enumerate}
		Furthermore, if these equivalent statements hold then $H$ coincides with the generator of $T(t)$.
	\end{thm}
	
	For the proof, we extend $T_d(t)$ to act as the identity on the orthogonal complement $V_d^\perp$ of $V_d$.
	This turns $T_d(t)$ into a strongly continuous contraction semigroup on $\H$, which would be wrong for the extension by $0$.
	Notice that this does not affect strong convergence: The difference of the extension by zero and the extension by the identity is the projection $P_d^\perp$ onto $V_d^\perp$ which vanishes strongly as $d\to\infty$.
	The proof is inspired by \cite[Thm.~2.3]{Duffield1992}:
	
	\begin{proof}
		
		For the proof, choose $z\in\mathbb{C}$ with $\Im z>0$.\\
		
		\ref{it:TKgroup} $\Rightarrow$ \ref{it:TKgen}:
		Our setting is a special case of the \emph{second Trotter-Kato approximation theorem} \cite[Thm.~4.9]{EN01}. Therefore, \ref{it:TKgroup} is equivalent to the statement that the resolvents $R(H_d;z)=(z-H_d)^{-1}$ are strongly convergent with $\slim_{d\to\infty} R(H_d;z)$ having dense range. For a sequence $(\ket{\psi_d})_d$ with $\ket{\psi_d}\in V_d$, we write $(\ket{\psi_d})_d\in\CCC$ if it converges in $\H$, i.e.,
		\begin{equation}
			\CCC = \Set*{(\ket{\psi_d})_d\subset \H \given \ket{\psi_d}\in V_d, \,\exists\,\ket{\psi} \in\H: \lim_{d\to\infty}\ket{\psi_d} = \ket{\psi}}.
		\end{equation}
		Furthermore, we write
		\begin{equation}
			\DDD := \Set[\big]{(\ket{\psi_d})_d\in\CCC \given (H_d\ket{\psi_d})_d\in\CCC}.
		\end{equation}
		Consider the resolvent operator $\RRR(z)$ which maps a sequence $(\ket{\psi_d})_d \subset\H$ to the sequence $(R(H_d;z)\ket{\psi_d})_d$.
		On one hand, by definition of $\DDD$, it holds that $((z-H_d)\ket{\psi_d})_d$ is a sequence in $\CCC$ if $(\ket{\psi_d})_d\in\DDD$, so
		\begin{equation}
			\DDD\subset \RRR(z)\CCC = \Set*{\rb*{R(H_d; z)\ket{\psi_d}}_d\given (\ket{\psi_d})_d\in\CCC}.
		\end{equation}
		On the other hand, since the $R(H_d; z)$ are uniformly bounded \cite[Thm.~II.1.10]{EN01} and strongly convergent in $d$, we have that $\rb*{R(H_d; z)\ket{\psi_d}}_d\in\CCC$ if $(\ket{\psi_d})_d\in\CCC$.
		But then $(\ket{\psi_d})_d\in\DDD$ due to $(z-H_d)R(H_d; z) = \1$ or, equivalently, $H_d R(H_d; z) = zR(H_d; z) - \1$.
		This shows
		\begin{equation}
			\DDD= \RRR(z)\CCC.
		\end{equation}
		For now, we denote the generator of the strong limit $T(t)$ by $A$ to avoid confusion with the operator $H$ defined in \ref{it:TKgen}.  Of course, we will show later that the two coincide.
		Note that by the Trotter-Kato Theorem, we have $\lim_d R(H_d;z) = R(A;z)$.
		For $(\ket{\psi_d})_d\in\CCC$ with $\lim_d \ket{\psi_d} = \ket{\psi}$, we have
		\begin{align}
			\lim_{d\to\infty} H_dR(H_d; z) \ket{\psi_d} &= \lim_{d\to\infty} \rb[\big]{z R(H_d; z)\ket{\psi_d} - \ket{\psi_d} } \nonumber\\
			&= z R(A; z)\ket{\psi}  - \ket{\psi} \nonumber\\
			&= A R(A; z)\ket{\psi}.
		\end{align}
		Since $\DDD=\RRR(z)\CCC$, this proves that $\lim_d H_d \ket{\phi_d} =A\rb[\big]{\lim_d \ket{\phi_d}}$ for all $(\ket{\phi_d})_d\in\DDD$. In particular, this shows the well-definedness of $H$ and that $A$ is an extension of $H$.
		Finally, we prove $H=A$ via
		\begin{equation}
			\DD(H) =
			\Set*{\lim_{d\to\infty}  \ket{\psi_d}\given (\ket{\psi_d})_d\in\DDD}
			=\Set*{\lim_d R(H_d;z) \ket{\phi_d} \given (\ket{\phi_d})_d \in\CCC}= R(A; z)\H = \D(A),
		\end{equation}
		where the third equality uses that $\lim_d R(H_d;z) \ket{\phi_d} = R(A;z) (\lim_d \ket{\phi_d})$ and that any $\ket\psi\in\H$ is the limit of a sequence in $\CCC$.
		
		\ref{it:TKgen} $\Rightarrow$ \ref{it:TKgroup}: Due to uniform boundedness, we only have to show strong convergence on the  dense subset $(z-H)\dom H \subset\H$.
		Let $\ket{\phi} \in \dom H$ and put $\ket{\psi} = (z-H)\ket{\phi} \in \H$. We know that $\ket{\psi} = \lim_{d\to\infty} (z-H_d) \ket{\phi}$.
		This implies that
		\begin{align}
			\|R(H_d;z)\ket{\psi} -  & R(H;z)\ket{\psi}\|\nonumber\\
			&\leq \norm{R(H_d;z)(\ket{\psi}-(z-H_d)\ket{\phi})} +\norm{R(H_d;z)(z-H_d)\ket{\phi} -\ket{\phi}}
		\end{align}
		goes to zero as $d\to\infty$: The first summand on the right-hand side is dominated by $\abs{z}^{-1} \norm{\ket{\psi}-(z-H_d)\ket{\phi}}\to0$. For the second summand, $R(H_d;z) (z-H_d)\ket{\phi} = \ket{\phi}$ for sufficiently large $d$. This shows that the resolvents $R(H_d;z)$ are strongly convergent. By construction, their limit $R(z)=\slim_{d\to\infty} R(H_d;z)$ has dense range. However, as already pointed out before, this is equivalent to \ref{it:TKgroup} due to the second Trotter-Kato approximation theorem \cite[Thm.~4.9]{EN01}.
	\end{proof}

	In Thm.~\ref{thm:TKAT}, we started with a sequence of generators $H_d$ and defined $H$ as a certain limit.
	However, for the proof of Thm.~\ref{thm:main}, we also need to look at the reversed setting, i.e.\ start from $H$ and define the $H_d$ as its finite-dimensional approximations. This scenario is elaborated in the following lemma.
	
	\begin{lem}\label{thm:strconv}
		If $\V\subset\D(H)$ is a core of $H$, then
		it holds that $P_dHP_d=:H_d\to H$ strongly on $\V$. Furthermore, we have $T_d(t)\to T(t)$ strongly on all of $\H$.
		\begin{proof}
			Let $\ket{\psi}\in\V$.
			Then there exists an $\tilde{d}\in\NN$ such that $\ket{\psi}\in V_{\tilde{d}}$.
			Hence, for $d\ge \tilde{d}$ it holds that $\ket{\psi_d} = P_d\ket{\psi} = \ket{\psi}$ and, thus, $H_d\ket{\psi_d} = P_dH\ket{\psi}\to H\ket{\psi}$.
			By the \emph{first Trotter-Kato approximation theorem} \cite[Thm.~4.8]{EN01}, the assertion follows.
		\end{proof}
	\end{lem}
	
	Before we turn towards strong convergence of the Trotter product, we provide the proof of Prop.~\ref{thm:uniformTrotter}.
	It also introduces some of the notation and techniques that we will employ later on.
	For the proof, recall Eq.~\eqref{eq:uniform_error_bound_beta}.
	\begin{proof}[Proof of Prop.~\ref{thm:uniformTrotter}]
		Let $H^{(i)}:\D(H^{(i)})\to\H$, $i=1,2$, be Hamiltonians such that $\V=\bigcup_d V_d\subset \D(H^{(1)})\cap\D(H^{(2)})$.
		As before, $H\supind i_d = P_dH\supind iP_d$.
		Denote $X_d(t) = \rme^{-\rmi tH_d^{(1)}}\rme^{-\rmi tH_d^{(2)}}$ and $U_d(t) = \rme^{-\rmi t(H_d^{(1)}+H_d^{(2)})}$ and by $X(t)$ and $U(t)$ the analogous operators on the infinite dimensional space.
		Assume that there exists $c_n = c_n(t) > 0$ such that $\beta_d^{(n)}(t) < c_n$ with $c_n\to 0$ if $n\to\infty$.
		Let $\psi\in\V$, w.l.o.g. $\psi\in V_d$, and let $\tilde{d}>d$.
		Then $P_{\tilde{d}}\psi = \psi$, so
		\begin{align}
			\norm{(X(t/n)^n - U(t))\psi}
			&\le \norm{(X_{\tilde{d}}(t/n)^n - U_{\tilde{d}}(t))\psi} + \norm{(\1-P_{\tilde{d}})(X(t/n)^n - U(t))\psi}\nonumber\\
			&\le c_n\norm{\psi} + \norm{(\1-P_{\tilde{d}})(X(t/n)^n - U(t))\psi}.
		\end{align}
		Taking $\tilde{d}\to\infty$, we thus obtain $\norm{(X(t/n)^n - U(t))\psi}\le c_n\norm{\psi}$.
		This shows uniform convergence on $\V$.
		Approximating $\psi\in\H$ by elements in $\V$ and using the triangle inequality and unitarity of $X(t)$ and $U(t)$ then shows $\norm{X(t/n)^n - U(t)}_\infty < c_n\to 0$.
	\end{proof}
	
	Recall the definitions in \eqref{eq:strong_Trotter_error_appendix} and \eqref{eq:Trotter_error_infinite}.
	We are now ready to present and prove:
	
	\begin{thm}\label{thm:main}
		Let $H^{(i)}:\D(H^{(i)})\to\H$, $i=1,2$, be generators of two strongly continuous contraction semigroups $T^{(i)}(t)$, $i=1,2$. Assume that $\V=\bigcup_d V_d\subset \D(H^{(1)})\cap\D(H^{(2)})$ is a common core of $H^{(1)}$ and $H^{(2)}$.
		Denote by $T_d(t)$ the strongly continuous contraction semigroup on $V_d$ generated by $H\supind1_d + H\supind2_d$.
		Then
		\begin{enumerate}[(1)]
			\item\label{it:main_somet}
			If for some $\ket{\psi}\in\V$, $t>0$ one has $b^{(n)}(\ket{\psi} ; t)\to 0$ as $n\to\infty$, then both limits $\lim_{d\to\infty}T_d(t)\ket{\psi}$ and $\lim_{n\to\infty}\rb*{T\supind1(t/n)T\supind2(t/n)}^n\ket{\psi}$ exist in $\H$ and coincide.
			\item\label{it:main_post}
			If for all $t\geq0$ and all $\ket{\psi}\in\V$ one has $b^{(n)}(\ket{\psi};t) \to 0$, then there is a strongly continuous contraction semigroup $T(t)$ on $\H$ with $T(t)\psi = \lim_d T_d(t)\ket \psi$, $\ket\psi\in\V$.
			The generator $H$ of $T(t)$ is an extension of $H\supind 1|_\V + H \supind 2|_\V$.
			The Trotter product of $T\supind1$ and $T\supind2$ converges strongly to $T(t)$.
		\end{enumerate}
		\begin{proof}
			We use the notations $X_d(t) = T\supind 1_d(t)T\supind 2_d(t)$ and $X(t)= T\supind1(t)T\supind2(t)$.
			
			\ref{it:main_somet}:
			We have to prove two statements. First, we prove that for the series of finite-dimensional approximations, $T_d(t)\ket{\psi}$ converges to a well-defined vector $\ket{\phi}$ as the level of truncation goes to infinity, $d\to\infty$. Second, we show that for the Trotter product, $X(t/n)^n\ket{\psi}$ converges to exactly this vector $\ket{\phi}$ in the Trotter limit $n\to\infty$. Let us start with the first assertion assuming $\ket{\psi}\in\V$. By the triangle inequality, it holds that
			\begin{equation}
				\norm*{\rb*{T_d(t) - T_{\tilde{d}}(t)}\ket{\psi}} \le b_d\supind n(\ket{\psi} ; t) + \norm*{\rb*{X_d(t/n)^n - X_{\tilde{d}}(t/n)^n}\ket{\psi}} + b_{\tilde{d}}\supind n(\ket{\psi} ; t),
			\end{equation}
			Due to Lemma~\ref{thm:strconv}, we have $\lim_{\tilde{d}\to\infty}\lim_{d\to\infty}\norm*{\rb*{X_d(t/n)^n - X_{\tilde{d}}(t/n)^n}\ket{\psi}} = 0$, and therefore
			\begin{equation}
				\limsup_{d\to\infty}\limsup_{\tilde{d}\to\infty} \norm*{\rb*{T_d(t) - T_{\tilde{d}}(t)}\ket{\psi}} \le 2b^{(n)}(\ket{\psi} ; t)\xrightarrow{n\to\infty} 0.
			\end{equation}
			Hence, $T_d(t)\ket\psi$ forms a Cauchy sequence and thus converges to some $\ket\phi\in\H$.
			For the second assertion, let $\ket{\psi}\in\V$ and use the triangle inequality again in order to receive
			\begin{align}
				\norm*{X(t/n)^n\ket{\psi} - \ket{\phi}}
				\le & \norm*{\rb*{X(t/n)^n - X_d(t/n)^n}\ket{\psi}}\nonumber\\
				&+ b_d\supind n(\ket\psi; t)\nonumber\\
				&+ \norm*{T_d(t) \ket{\psi} - \ket\phi}.
			\end{align}
			By convergence of $T_d(t)\ket{\psi}$ and, therefore, $X_d(t)\ket{\psi}$ and by taking the $\limsup_{d\to\infty}$, we obtain
			\begin{equation}
				\norm*{X(t/n)^n\ket{\psi} - \ket{\phi}} \le b^{(n)}(\ket{\psi}; t)\xrightarrow{n\to\infty} 0,
			\end{equation}
			which proves the claim.
			
			\ref{it:main_post}:
			Again, two claims are to be proved. For the first claim, we have to show that $T(t)$ is a strongly continuous contraction semigroup on $\H$. For the second claim, we have to show that the generator $H$ of $T(t)$ is an extension of $H\supind 1|_\V + H \supind 2|_\V$, i.e.\ the sum of the individual generators of $T^{(1)}(t)$ and $T^{(2)}(t)$. For the first claim, notice that by (1), $T(t)$ is the strong limit of the semigroups $T_d(t)$. Hence, it is a semigroup itself.
			For strong continuity, let $\ket{\psi}\in\V$.
			Then for $d$ large enough we have that $\ket{\psi}\in V_d$ and it holds that \cite[Lemma~1.3]{EN01}
			\begin{align}
				\|T_d(t)\ket{\psi} - \ket{\psi}\| &= \norm*{\int_0^t T_d(s)(H\supind 1_d + H\supind 2_d)\ket{\psi}\d s} \nonumber\\
				&\le \int_0^t\|(H_d\supind 1 + H_d\supind 2)\ket{\psi}\|\d s\nonumber\\
				&= t\|P_d(H\supind 1 + H\supind 2)\ket{\psi}\|.
			\end{align}
			By continuity of $\norm{\placeholder}$ we have
			\begin{equation}
				\|T(t)\ket{\psi} - \ket{\psi}\| = \lim_{d\to\infty} \|T_d(t)\ket{\psi} - \ket{\psi}\|\le t\|(H\supind 1 + H\supind 2)\ket{\psi}\|\xrightarrow{t\to0} 0.
			\end{equation}
			By density of $\V$ in $\H$ and, again, uniform continuity, strong continuity of $T(t)$ follows on all of $\H$.
			Now, let $\ket{\psi}\in\V\subset\D(H\supind 1)\cap\D(H\supind 2)$ and $\ket{\psi_d} = P_d\ket{\psi}$.
			Then as $d\to\infty$, $\ket{\psi_d}\to \ket{\psi}$ and $H\supind i_d\ket{\psi_d} \to H\supind i\ket{\psi}$, since for $d$ large enough $\ket{\psi_d} = \ket{\psi}$.
			Hence, $(H_d\supind 1 + H_d\supind 2)\ket{\psi_d}\to (H\supind 1 + H\supind 2)\ket{\psi}$, so $\ket{\psi}\in\D(H)$ and $H\ket{\psi} = (H\supind 1 + H\supind 2)\ket{\psi}$ by Thm.~\ref{thm:TKAT}~\ref{it:TKgen}.
			Thus $\conj{H\supind 1_{|\V} + H\supind 2_{|\V}} \subset H$.
		\end{proof}
	\end{thm}

	Now, our Main~Result~\ref{thm:mainthm} follows as a special case of Thm.~\ref{thm:main}:

	\begin{main_two}
		Let $H^{(1)}:\D(A^{(1)})\to\H$ and $H^{(2)}:\D(H^{(2)})\to\H$ be self-adjoint, hence generators of two strongly continuous unitary groups $T^{(i)}(t)$, $i=1,2$. Assume that $\V=\bigcup_d V_d\subset \D(H^{(1)})\cap\D(H^{(2)})$ is a common core of $H^{(1)}$ and $H^{(2)}$. If for all $t\in \RR$, $b^{(n)}(\ket{\psi};t) \to 0$ as $n\to\infty$, then
		\begin{equation}
			\rb*{T\supind1(t/n)T\supind2(t/n)}^n\ket{\psi} \xrightarrow{n\to\infty} T(t)\ket{\psi},
		\end{equation}
		where $T(t)$ is a strongly continuous unitary group, whose generator $H$ is a self-adjoint and agrees with $(H_1 + H_2)$ wherever both are defined.
		\begin{proof}
			Since $H\supind 1$ and $H\supind 2$ are self-adjoint,
			so are $H\supind i_d, i=1,2$ and $H_d$.
			All these operators generate unitary groups that we will denote by $U\supind i$, $U\supind i_d$ and $U_d$, respectively.
			Since $b^{(n)}(\ket{\psi}; t)\to 0$ as $n\to\infty$ for $t\ge 0$, we obtain a contraction semigroup $U(t)$, $t\ge 0$, generated by $H$.
			Now, $U(t)$ is isometric:
			$\|U(t)\ket{\psi}\| = \lim_{d\to\infty}\norm*{U_d(t)\ket{\psi}} = \|\ket{\psi}\|.$
			Hence, $H$ is symmetric as the generator of a strongly continuous isometry semigroup. To see this, let $\ket{\psi},\ket{\phi} \in \D(H)$ and differentiate $\langle\psi|\phi\rangle = \langle U(t)\psi | U(t)\phi\rangle$ at $t=0$.
			The same argument can be applied for $t\le 0$, in which case we obtain a semigroup generated by $-H^*$. Hence, $H^*$ is symmetric as well.
			Notice that an operator $A$ is symmetric if and only if (i) $A \subset A^*$, and (ii) $A \subset B$ implies $B^* \subset A^*$ \cite{ReedSimon1981}. Applying this to $H$ and $H^*$ gives $H \subset H^* \subset H^{**} = H$, which shows that $H$ is self-adjoint.
			In order to see that $H$ is equal to $(H\supind1+H\supind2)$ wherever both are defined, we consider the difference: Let $Y$ be the closure of the symmetric operator $(H - H\supind1+H\supind2)$, defined on the intersection of the domains. We know from Thm.~\ref{thm:main}~\ref{it:main_post} that $Y |_\V = 0$.
			But a closed operator which vanishes on a dense subset must vanish everywhere. To see this, note that since $Y$ is closed, we have that $Y$ is an extension of $\overline{(Y|_\V)}$.
			But $Y|_\V=0 |_\V$ and since the zero operator is bounded this means that $\overline{(Y|_\V)} =0$. As an extension of the zero operator, $Y$ is itself zero, i.e., $H = H\supind1+H\supind2$ wherever all are defined.
		\end{proof}
	\end{main_two}

	It remains to prove Eq.~\eqref{eq:error_bound_infinite} from the main text.
	\begin{proof}[Proof of Eq.~\eqref{eq:error_bound_infinite}]
		By the triangle inequality, $b^{(n)}(\ket{\psi};t)^2
		\leq \limsup_{d\to\infty}\sum_{j=1}^d \big|\langle j|\psi\rangle\big|^2 b_d^{(n)}(\ket{j};t)^2$.
		If $\ket{\psi} \in \V$, then there exists an $\tilde{d}\in\NN$, such that $\ket{\psi}\in V_{\tilde{d}}$ and therefore $\ket{\psi}=\sum_{j=1}^{\tilde{d}} \langle j|\psi\rangle\ket{j}$. As a consequence for $d\geq \tilde{d}$,
		$b_d^{(n)}(\ket{\psi};t)^2\leq\sum_{j=1}^{\tilde{d}} \big|\langle j|\psi\rangle\big|^2 b_d^{(n)}(\ket{j};t)^2$,
		and therefore,
		\begin{align}
			b^{(n)}(\ket{\psi};t)
			\leq \rb[\bigg]{\sum_{j=1}^{\tilde{d}} \big|\langle j | \psi \rangle\big|^2 b^{(n)}(\ket j;t)^2 }^{1/2},
		\end{align}
		which proves the statement.
	\end{proof}

	\hypertarget{app:example}{}
	\section{D. Example: The quantum harmonic oscillator}
	The bounds from our Main~Result~\ref{thm:finite_bounds} hold for arbitrary Hamiltonians acting on a finite-dimensional Hilbert space. In particular, we can use them for bounding strong Trotter errors for the finite-dimensional truncated operators, in which we are interested. For instance, consider the case where $H^{(1)}=\frac{1}{2}Q^2$ and $H^{(2)}=\frac{1}{2}P^2$, where where $Q:\mathcal{{D}}(Q)\rightarrow L^{2}(\mathbb{{R}})$
	is the position operator and $P:\mathcal{{D}}(P)\rightarrow L^{2}(\mathbb{{R}})$
	is the momentum operator. We can rewrite $Q$ and $P$ in terms of
	the creation operator $a^{\dagger}$ and annihilation operator $a$
	as follows
	\begin{align}
		Q & =\frac{1}{\sqrt{2}}\left(a+a^{\dagger}\right),\\
		P & =\frac{\mathrm{{i}}}{\sqrt{2}}\left(a^{\dagger}-a\right).
	\end{align}
	The operators $a$ and $a^{\dagger}$ satisfy the commutation relation
	$\left[a,a^{\dagger}\right]=\1$. In this particular Trotter scenario, we would like to implement the target Hamiltonian $H^\mathrm{osc}=\frac{1}{2}\left(Q^2+P^2\right)$ by switching between the time evolutions generated by $H^{(1)}=\frac{1}{2}Q^2$ and $H^{(2)}=\frac{1}{2}P^2$, respectively. $H^\mathrm{osc}$ is the quantum harmonic oscillator and is related to the \emph{number operator} $N=a^{\dagger}a$
	by $H^\mathrm{osc}=N+\frac{1}{2}$. The eigenstates $\left|m\right\rangle $ of $N$ according
	to the eigenvalue equation $N\left|m\right\rangle =m\left|m\right\rangle $
	are the Fock states.
	In turn, the $\ket{m}$ are eigenstates of $H^\mathrm{osc}$ as well and satisfy the
	eigenvalue equation $H^\mathrm{osc}\left|m\right\rangle =\left(m+\frac{1}{2}\right)\left|m\right\rangle $.
	Let us use our Main~Result~\ref{thm:finite_bounds} to compute a strong error bound for this Trotter problem for the corresponding finite-dimensional approximations.
	Obviously, we have
	\begin{align}
		Q^{2} & =\frac{1}{2}\left(2N+aa+a^{\dagger}a^{\dagger}+1\right),\\
		P^{2} & =\frac{1}{2}\left(2N-aa-a^{\dagger}a^{\dagger}+1\right)
	\end{align}
	and the Fock states are a common core of $H^{(1)}=\frac{1}{2}Q^2$ and $H^{(2)}=\frac{1}{2}P^2$.
	Define $R^{\pm}=\frac{1}{2}\left(2N\pm aa\pm a^{\dagger}a^{\dagger}+1\right)$,
	so that $R^{+}=Q^{2}$ and $R^{-}=P^{2}$. Then the action of $R^\pm$ onto a Fock state $\ket{m}$ with $m\geq2$ computes to
	\begin{equation}
		R^{\pm}\left|m\right\rangle =\left(m+\frac{1}{2}\right)\left|m\right\rangle \pm\frac{1}{2}\sqrt{m\left(m-1\right)}\left|m-2\right\rangle \pm\frac{1}{2}\sqrt{\left(m+1\right)\left(m+2\right)}\left|m+2\right\rangle .\label{eq:Rm}
	\end{equation}
	Analogously, for $\left(R^\pm\right)^2$ we obtain
	\begin{align}
		\left(R^{\pm}\right)^{2} =&\frac{1}{4}\left(2N\pm aa\pm a^{\dagger}a^{\dagger}+1\right)^{2}\nonumber\\
		=&\frac{1}{4}\left(4N^{2}\pm2Naa\pm2Na^{\dagger}a^{\dagger}+4N\pm2aaN+aaaa+aaa^{\dagger}a^{\dagger}\pm2aa\right.\nonumber\\&\left.\pm2a^{\dagger}a^{\dagger}N+a^{\dagger}a^{\dagger}aa+a^{\dagger}a^{\dagger}a^{\dagger}a^{\dagger}\pm2a^{\dagger}a^{\dagger}+1\right)\nonumber\\
		=&\frac{1}{4}\left(4N^{2}\pm4aaN\pm4a^{\dagger}a^{\dagger}N+2aaa^{\dagger}a^{\dagger}+aaaa+a^{\dagger}a^{\dagger}a^{\dagger}a^{\dagger}\mp2aa\pm6a^{\dagger}a^{\dagger}-1\right),
	\end{align}
	where we have used $\left[N,aa\right]=-2aa$, $\left[N,a^{\dagger}a^{\dagger}\right]=2a^{\dagger}a^{\dagger}$
	and $\left[aa,a^{\dagger}a^{\dagger}\right]=4N+2$. Hence, by acting on a Fock state $\ket{m}$ with $m\geq4$,
	\begin{align}
		\left(R^{\pm}\right)^{2}\left|m\right\rangle  =&\left(N^{2}\pm aaN\pm a^{\dagger}a^{\dagger}N+\frac{1}{2}aaa^{\dagger}a^{\dagger}+\frac{1}{4}aaaa+\frac{1}{4}a^{\dagger}a^{\dagger}a^{\dagger}a^{\dagger}\mp\frac{1}{2}aa\pm\frac{3}{2}a^{\dagger}a^{\dagger}-\frac{1}{4}\right)\left|m\right\rangle \nonumber\\
		=&m^{2}\left|m\right\rangle \pm m\sqrt{m\left(m-1\right)}\left|m-2\right\rangle \pm m\sqrt{\left(m+1\right)\left(m+2\right)}\left|m+2\right\rangle\nonumber\\& +\frac{1}{2}\left(m+1\right)\left(m+2\right)\left|m\right\rangle +\frac{1}{4}\sqrt{m\left(m-1\right)\left(m-2\right)\left(m-3\right)}\left|m-4\right\rangle\nonumber\\& +\frac{1}{4}\sqrt{\left(m+1\right)\left(m+2\right)\left(m+3\right)\left(m+4\right)}\left|m+4\right\rangle \mp\frac{1}{2}\sqrt{m\left(m-1\right)}\left|m-2\right\rangle\nonumber\\& \pm\frac{3}{2}\sqrt{\left(m+1\right)\left(m+2\right)}\left|m+2\right\rangle -\frac{1}{4}\left|m\right\rangle \nonumber\\
		=&\frac{3}{4}\left(2m\left(m+1\right)+1\right)\left|m\right\rangle \pm\frac{1}{2}\left(2m-1\right)\sqrt{m\left(m-1\right)}\left|m-2\right\rangle\nonumber\\& \pm\frac{1}{2}\left(2m+3\right)\sqrt{\left(m+1\right)\left(m+2\right)}\left|m+2\right\rangle +\frac{1}{4}\sqrt{m\left(m-1\right)\left(m-2\right)\left(m-3\right)}\left|m-4\right\rangle\nonumber\\& +\frac{1}{4}\sqrt{\left(m+1\right)\left(m+2\right)\left(m+3\right)\left(m+4\right)}\left|m+4\right\rangle.
	\end{align}
	In order to apply our Main~Result~\ref{thm:finite_bounds}, we have to compute the norm of
	\begin{align}
		\left(\left(R^{\pm}\right)^{2}-2\left(m+\frac{1}{2}\right)R^{\pm}+\left(m+\frac{1}{2}\right)^{2}\right)\left|m\right\rangle  =&\frac{3}{4}\left(2m\left(m+1\right)+1\right)\left|m\right\rangle \pm\frac{1}{2}\left(2m-1\right)\sqrt{m\left(m-1\right)}\left|m-2\right\rangle\nonumber\\& \pm\frac{1}{2}\left(2m+3\right)\sqrt{\left(m+1\right)\left(m+2\right)}\left|m+2\right\rangle\nonumber\\& +\frac{1}{4}\sqrt{m\left(m-1\right)\left(m-2\right)\left(m-3\right)}\left|m-4\right\rangle\nonumber\\& +\frac{1}{4}\sqrt{\left(m+1\right)\left(m+2\right)\left(m+3\right)\left(m+4\right)}\left|m+4\right\rangle \nonumber\\&
		-\left(m+\frac{1}{2}\right)^{2}\left|m\right\rangle \mp\left(m+\frac{1}{2}\right)\sqrt{m\left(m-1\right)}\left|m-2\right\rangle\nonumber\\& \mp\left(m+\frac{1}{2}\right)\sqrt{\left(m+1\right)\left(m+2\right)}\left|m+2\right\rangle \nonumber\\
		=&\frac{1}{2}\left(m^{2}+m+1\right)\left|m\right\rangle \mp\sqrt{m\left(m-1\right)}\left|m-2\right\rangle\nonumber\\& \pm\sqrt{\left(m+1\right)\left(m+2\right)}\left|m+2\right\rangle +\frac{1}{4}\sqrt{m\left(m-1\right)\left(m-2\right)\left(m-3\right)}\left|m-4\right\rangle\nonumber\\& +\frac{1}{4}\sqrt{\left(m+1\right)\left(m+2\right)\left(m+3\right)\left(m+4\right)}\left|m+4\right\rangle .\label{eq:R2Fock}
	\end{align}
	Since the $\left\{ \left|m\right\rangle \right\} $ form an orthonormal
	basis,
	\begin{align}
		\left\Vert\left(\left(R^{\pm}\right)^{2}-2\left(m+\frac{1}{2}\right)R^{\pm}+\left(m+\frac{1}{2}\right)^{2}\right)\left|m\right\rangle\right\Vert^2=&\frac{3}{8} \left(m (m+1) \left(m^2+m+14\right)+10\right).\label{eq:error_term_HO}
	\end{align}
	We now truncate the operators at dimension $d$ in the Fock basis. That is, we use the orthogonal
	projector $P_{d}=\sum_{m=0}^{d-1}\left|m\right\rangle \left\langle m\right|$
	as a truncation scheme. Denote the truncated vectors and operators with a subscript $d$ and assume that $d\geq m+4$. Then 
	\begin{equation}
		\left\Vert \left(U_{d}(t)-W_{d}^{(n)}(t)\right)\left|m_{d}\right\rangle \right\Vert \leq\frac{t^{2}}{4n}\sqrt{\frac{3}{8} \left(m (m+1) \left(m^2+m+14\right)+10\right)},\label{eq:Trotter_error_HO}
	\end{equation}
	which is Eq.~\eqref{Trotter_X2P2} of the main text. The assumption $d\geq m+4$ is made in order to ensure that  $(R_{d}^{\pm})^2$ (and $R_{d}^{\pm}$)
	acts correctly on the truncated Fock states $\left|m_{d}\right\rangle $:
	Since $(R^{\pm})^2$ involves terms up to fourth order in $a^{\dagger}$,
	we have to make sure that all vectors in the linear combination of
	$(R^{\pm})^2\left|m\right\rangle $ in Eq.~\eqref{eq:R2Fock} are actually again vectors in the truncated Hilbert space $V_{d}$.
	If we drop this assumption and consider $m\leq d<m+4$ instead, all
	terms in $(R^{\pm})^2\left|m\right\rangle$, which give Fock states $\left|m+j\right\rangle $
	with $j\geq d-m+1$, will not contribute to $(R_{d}^{\pm})^2\left|m_{d}\right\rangle $.
	In this case, Eq.~\eqref{eq:error_term_HO}
	will become smaller and Eq.~\eqref{eq:Trotter_error_HO} will still give an upper bound.
	Similarly, for the validity of Eq.~\eqref{eq:Rm}--\eqref{eq:error_term_HO}, we made the assumption $m\geq4$. However, our bound in Eq.~\eqref{eq:Trotter_error_HO} stays valid even for Fock states $\ket{m}$ with $m<4$. In this case, all terms for which $m-4<0$ vanish in the derivation of the bound. Nevertheless, adding these additional terms only increases the norm so that Eq.~\eqref{eq:Trotter_error_HO} is still an upper bound.
	This bound is independent of the truncation dimension $d$, which
	shows that the full infinite-dimensional Trotter problem $H^{(1)}=\frac{1}{2}Q^{2}$,
	$H^{(2)}=\frac{1}{2}P^{2}$ converges strongly to $H^\mathrm{osc}=\frac{1}{2}\left(X^{2}+P^{2}\right)$.
	Furthermore, since $Q^{2}$ and $P^{2}$ are self-adjoint operators
	when acting on the Fock space $\mathcal{F}=\overline{\mathrm{span}\left\{ \left|m\right\rangle \right\}} $,
	$H^\mathrm{osc}=\frac{1}{2}\left(Q^{2}+P^{2}\right)$ is a self-adjoint operator
	on the Fock space.

\end{document}